\documentclass[11pt,twoside,letter]{article}

\usepackage{amsmath, amssymb, amsthm}
\usepackage[initials,short-months]{amsrefs}
\usepackage{tikz}

\theoremstyle{plain}
\newtheorem{theorem}{Theorem}[section]
\newtheorem{lemma}[theorem]{Lemma}
\newtheorem{corollary}[theorem]{Corollary}

\theoremstyle{definition}
\newtheorem{definition}{Definition}
\newtheorem{example}{Example}
\newtheorem*{condA}{Condition~A}
\newtheorem*{condB}{Condition~B}

\theoremstyle{remark}
\newtheorem{remark}{Remark}

\def\Hx{ \widehat{x} }
\def\Hy{ \widehat{y} }
\def\Hu{ \widehat{u} }
\def\Hv{ \widehat{v} }
\def\Hbv{ \widehat{\mathbf{v}}}
\def\Hbx{ \widehat{\mathbf{x}}}
\def\d{ \mathrm{d} }						
\def\I{ \mathrm{i} }						
\def\E{ \mathrm{e} }						
\def\T{ \mathrm{T} }						
\def\NN{ \mathbb{N} }						
\def\ZN{ \mathbb{Z} }						
\def\RN{ \mathbb{R} }						
\def\CN{ \mathbb{C} }						
\def\TT{ \mathbb{T} }                             
\def\L{ \mathcal{L} }					     

\def\ourB{ \mathcal{B}_{T/2}}
\def\ourSpace{ \mathcal{Y}_{T/2} }
\def\bsc{ \boldsymbol{c} }					
\def\balpha{ \boldsymbol{\alpha} }
\def\blambda{ \boldsymbol{\lambda} }
\def\bc{ \mathbf{c} }						
\def\bx{ \mathbf{x} }						

\DeclareMathOperator*{\sinc}{sinc}							
\newcommand{\Op}[1]{\mathrm{#1}}

\begin{document}

\title{Phaseless Signal Recovery in Infinite Dimensional Spaces using Structured Modulations
\thanks{Preprint accepted for publication in the \emph{Journal of Fourier Analysis and Applications}.}
}

\author{
Volker Pohl\\
\small Lehrstuhl f{\"u}r Theoretische Informationstechnik\\
\small Technische Universit{\"a}t M{\"u}nchen, Arcisstrasse 21, 80290 M{\"u}nchen, Germany\\
\small e-mail: volker.pohl@tum.de\\
\\
Fanny Yang \\
\small Department of Electrical Engineering and Computer Science\\
\small University of California-Berkeley,  Berkeley, CA~94720, USA\\
\small e-mail: fanny-yang@berkeley.edu\\
\\
Holger Boche\\
\small Lehrstuhl f{\"u}r Theoretische Informationstechnik\\
\small Technische Universit{\"a}t M{\"u}nchen, Arcisstrasse 21, 80290 M{\"u}nchen, Germany\\
\small e-mail: boche@tum.de}

\date{July 10, 2014}

\maketitle

\markboth{\footnotesize \rm \hfill V.~Pohl, F.~Yang, H.~Boche}
{\footnotesize \rm Phaseless Signal Recovery in Infinite Dimensional Spaces\hfill}

\begin{abstract}
This paper considers the recovery of continuous signals in infinite dimensional spaces from the magnitude of their frequency samples.
It proposes a sampling scheme which involves a combination of oversampling and modulations with complex exponentials. 
Sufficient conditions are given such that almost every signal with compact support can be reconstructed up to a unimodular constant using only its magnitude samples in the frequency domain. Finally it is shown that an average sampling rate of four times the Nyquist rate is enough to reconstruct almost every time-limited signal.

\vspace{5mm}
\vspace{3mm}
\end{abstract}

\maketitle

\newpage
\section{Introduction}
\label{sec:Intro}

In many different applications of engineering and science, the signal of interest needs to reconstructed from magnitude measurements only, which is well-known as the \emph{phase retrieval problem}.
This originates from the fact that detectors can often only record the squared modulus of an electromagnetic wave so that phase information is lost.
One of its most important applications is during the image reconstruction stage in X-ray crystallography, widely used in material and biological science \cites{Millane_90,Miao_XRay_08}.
There the phase of radiation, scattered from an object, carries information about the surface or the inner structure of this object.
The phase retrieval problem also appears in numerous other areas of imaging science including electron microscopy, astronomical imaging \cite{Fienup_93}, diffraction imaging, X-ray tomography, to mention only a few.
Furthermore it plays an important role in other fields such as speech processing \cite{Rabiner_SpeechRec}, radar \cite{Jaming_Radar10}, signal theory \cite{Oppenheim_Phase_80}, quantum communication \cite{Finkelstein_QuantumCom04}.

The underlying problem lies in the fact that in general, magnitude and phase of a signal can be chosen independently.
Therefore signal recovery from magnitude measurements is only possible if it is combined with additional information about the signal.
If, for example, the z-transform of the signal satisfies certain sufficient conditions \cite{Oppenheim_Phase_80}, then the signal is completely determined by its magnitude.
Also, if the signal is known to be causal or band-limited, then the logarithm of the magnitude and its phase are related by the Hilbert transform \cites{Burge_PhaseProblem76,Ross_PhaseProblem78}.
However, it was shown in \cite{Boche_Pohl_IEEE_IT_InterpolatedData} that there exists no stable linear method to approximate the Hilbert transform from finitely many discrete measurements, even if the number of measurements is arbitrarily large.

When no or only little a priori knowledge is available on the original signal, one can still successfully perform reconstruction from magnitude measurements by taking several measurements of the same object under slightly different conditions.
To this end various methods were proposed, such as
a distorted-object approach by which the Fresnel diffraction pattern is measured at different distances \cite{Xiao_DistortedObject05},
the usage of aperture-plane modulation \cites{Zhang_ApatureMod07, Falldorf_SLM10}, 
or the recording of several fractional Fourier transforms \cite{Jaming_Radar10}.
Having obtained several measurements, signal recovery is mainly performed by iterative alternating projection algorithms \cite{Fienup1982_PhaseRetrieval}.
Although these algorithms are usually easy to implement, they often come with problems in terms of convergence (see, e.g., \cites{Bauschke02,Marchesini_07}), which strongly depends on specific signal constraints.
Moreover, in connection with the described multiple measurement methods, there seems to be no systematic approach to design the different measurements such that iterative algorithms will definitely converge to the correct signal.

More recently, analytic investigations on the phase retrieval problem resulted in sufficient conditions on the number of measurements for perfect signal recovery.
In \cites{Balan1, Bodmann_StablePR2014} it was found that for an $N$-dimensional space slightly less than $4 N$ suitably chosen amplitude measurements are sufficient. 
However, no corresponding reconstruction algorithm was given which achieves this bound. 
The closest result was presented in \cite{Balan_Painless}, where the authors proposed a method which guarantees signal recovery from amplitude measurements requiring $N^2$ measurements. This fundamentally limits its practicability to low dimensional spaces.
Ideas of sparse signal representation and convex optimization were applied in \cites{Vetterli2011_PRbeyondFienup,CandesEldar_PhaseRetrieval} to allow for lower computational complexity.

All of the above approaches address finite dimensional signals. A natural question is whether similar results can be obtained for continuous signals in infinite dimensional spaces.
In \cite{Thakur2011} it was shown that real valued band-limited signals are completely determined simply by their magnitude samples taken at twice the Nyquist rate.
It is not clear whether a similar result holds for complex valued signals, 
since results for finite dimensional spaces indicate that oversampling alone may not be sufficient. In \cites{Balan1,Balan_Painless,CandesEldar_PhaseRetrieval}, the particular choice of measurement vectors was the key to enabling signal recovery.
The first attempt on complex valued time-limited $\L^{2}$-signals was made in \cites{Yang_SampTA13} where recovery was guaranteed given 
specific amplitude measurements taken at four times the Nyquist rate. It provides a reconstruction algorithm which incorporates ideas from finite dimensional spaces in \cite{Balan_Painless} and was inspired by the structured illuminations frequently used in optics \cites{Xiao_DistortedObject05,Zhang_ApatureMod07,Falldorf_SLM10,CandesEldar_PhaseRetrieval}.

The present work extends the result of \cite{Yang_SampTA13} to larger signal spaces, namely to compactly supported signals in $\L^{p}$ with $1 \leq p \leq \infty$.
As in \cite{Yang_SampTA13}, a bank of modulators is applied before the intensity measurement. We derive conditions on the sampling rate and the modulators such that every signal in a countable intersection of open dense subsets can be reconstructed, up to a unimodular constant, from samples taken at a rate which may be arbitrarily close to four times the Nyquist rate.

The paper is organized as follows:
Basic notations and some preliminary results are presented in Section~\ref{sec:Notation}, whereas sampling and reconstruction in Bernstein spaces are recaptured in Section~\ref{sec:InterpolSeq}.
Our sampling setup is described in Section~\ref{sec:Measurement}. In particular two conditions on the sampling system are introduced which will enable signal recovery.
Section~\ref{sec:MainResult} will prove that these two conditions are indeed sufficient for signal reconstruction from magnitude measurements, with the exception of a set of signals of first category.
Section~\ref{sec:SubSpaces} shows that by a small change of the sampling system and under mild restrictions on the signal space, even this limitation can be avoided.
The paper closes with a short discussion in Section~\ref{sec:Summary}.

\section{Notations and Preliminaries}
\label{sec:Notation}

\paragraph{General Notations}
Let $\mathbb{S} \subseteq \RN$ be an arbitrary interval of the real axis $\RN$.
For any exponent $1 \leq p \leq \infty$ we write $\L^{p}(\mathbb{S})$ for the usual Lebesgue space on $\mathbb{S}$ and $\ell^{p}$ for the common sequence spaces.
The exponent $p'$ which satisfies $1/p + 1/p' = 1$ is the \emph{conjugate exponent} of $p$ and by convention $p'=\infty$ is the conjugate exponent of $p=1$.
In particular, $\L^{2}(\mathbb{S})$ is the Hilbert space of square integrable functions on $\mathbb{S}$ with the usual inner product
$\left\langle x,y \right\rangle_{\L^{2}(\mathbb{S})} = \int_{\mathbb{S}} x(\theta)\, \overline{y(\theta)}\, \d\theta$
where the bar denotes the complex conjugate and the integral is taken with respect to the Lebesgue measure. 
The set of continuous functions on $\RN$ equipped with the supremum norm is denoted by $C(\RN)$ and $C_{0}(\RN)$ stands for the subset of all $x \in C(\RN)$ which vanish at infinity.
As usual $\CN^{N}$ stands for the $N$-dimensional Euclidean Hilbert space.

For $1 \leq p \leq \infty$ we denote by $\mathcal{B}^{p}_{\sigma}$ the \emph{Bernstein space} of all entire functions of exponential type $\leq \sigma$ whose restriction to $\RN$ belongs to $\L^{p}(\RN)$.
The norm in $\mathcal{B}^{p}_{\sigma}$ is defined by the $\L^{p}$-norm on $\RN$.
The theorem of Plancherel-P{\'o}lya implies (see, e.g. \cite[Chap.~20.1]{Levin1997_Lectures}) that there is a constant $C = C(p,\sigma)$ such that for all $x \in \mathcal{B}^{p}_{\sigma}$
\begin{equation}
\label{equ:boundedBp}
	|x(\xi + \I \eta)| \leq C\, \|x\|_{\L^{p}(\RN)}\, \E^{\sigma |\eta|}
	\quad\text{for all}\ \xi,\eta \in\RN\;.
\end{equation}
Consequently, $\mathcal{B}^{p}_{\sigma} \subset \mathcal{B}^{\infty}_{\sigma}$ for all $1 \leq p < \infty$ and convergence in $\mathcal{B}^{p}_{\sigma}$ implies uniform convergence in each horizontal strip of the complex plane.

\paragraph{Fourier-Laplace transform and Paley-Wiener theorem}

The \emph{Schwartz space} $\mathcal{S}$ of rapidly decreasing functions on $\RN$ is assumed to be equipped with the usual locally convex topology (see, e.g., \cites{Rudin_FktAnalysis,Hoermander_LinDifOp}).
The elements of its dual space $\mathcal{S}'$ are \emph{tempered distributions}.
If $x \in \mathcal{S}'$ can be represented by a locally integrable function $\tilde{x} \in \L^{1}_{\mathrm{loc}}(\RN)$ such that
$x(\phi) = \int_{\RN} \tilde{x}(t)\, \phi(t)\, \d t$
then $x$ is called \emph{regular}. 

For any $\phi \in \mathcal{S}$, its \emph{Fourier transform} is the function $\widehat{\phi} : \RN \to \CN$ defined by
\begin{eqnarray}
\label{equ:FourierTransform}
	& \widehat{\phi}(\omega)
	= (\mathcal{F} \phi)(\omega)
	= \int_{\RN} \phi(t)\, \E^{-\I \omega t}\, \d t\;,
	\quad \omega\in\RN\;.
\end{eqnarray} 
It is known (see, e.g., \cite{Vladimirov}) that $\mathcal{F} : \mathcal{S} \to \mathcal{S}$ is a continuous, bijective mapping
and that $\mathcal{F}$ can be extended to any $x\in\mathcal{S}'$ by defining its Fourier transform as $\Hx(\phi) = (\mathcal{F}x)(\phi) = x(\widehat{\phi})$ for all $\phi \in \mathcal{S}$.
Then $\mathcal{F} : \mathcal{S}' \to \mathcal{S}'$ is bijective and continuous.

Of particular importance are the subspaces $\L^p(\RN) \subset \mathcal{S}'$ of regular distributions. 
The Hausdorff-Young inequality establishes that for $1 \leq p' \leq 2$ the Fourier transform $\mathcal{F} : \L^{p'}(\RN) \to \L^{p}(\RN)$ is a bounded linear map with
\begin{equation}
\label{equ:boundedFTrafo}
	\|\mathcal{F} x\|_{\L^{p}} \leq \tfrac{1}{2\pi}\, \|x\|_{\L^{p'}}\;.
\end{equation}
However, apart from the case $p'=2$, the map $\mathcal{F} : \L^{p'}(\RN) \to \L^{p}(\RN)$ is not surjective
and there exist functions $\Hx \in \L^{p}(\RN)$ whose inverse Fourier transform $x = \mathcal{F}^{-1} \Hx$ are tempered distributions $x \in \mathcal{S}'$ which are not regular.

The Fourier transform \eqref{equ:FourierTransform} can be extended to the whole complex plane $\CN$ by replacing the variable $\omega \in \RN$ by $z \in \CN$.
Then $\Hx(z) = (\mathcal{F} x)(z)$ is an entire function. Similarly, one can define the so-called \emph{Fourier-Laplace transform} on $\mathcal{S}'$ by $\Hx(z) = x(e_{-z})$, where $e_{-z}(t) := \E^{-\I z t} \in C^{\infty}(\RN)$ for every $x \in \mathcal{S}'$ with compact support \cites{Hoermander_LinDifOp,Rudin_FktAnalysis}.
Then $\Hx$ is an entire function and its restriction to $\RN$ is the usual Fourier transform \eqref{equ:FourierTransform}.
The following theorem provides a relation between entire functions of exponential type and the Fourier-Laplace transform of tempered distributions with compact support \cites{Rudin_FktAnalysis,Hoermander_LinDifOp}.

\begin{theorem}[Paley-Wiener]
\label{thm:PaleyWiener}
An entire function $\Hx(z)$ is the Fourier-Laplace transform of a (tempered) distribution with compact support in the interval $[-\sigma,\sigma]$
if and only if there exists a constant $C$ and an $N \in \NN$ such that
\begin{equation*}
	|\Hx(\xi + \I\eta)|
	\leq C\, \left( 1 + |\xi + \I\eta| \right)^{N} \E^{\sigma |\eta |}\;.
\end{equation*}
\end{theorem}
In particular, every entire function of exponential type $\leq \sigma$ is the Fourier-Laplace transform of a (tempered) distribution with compact support contained in $[-\sigma,\sigma]$.

\section{Sampling and Reconstruction in Bernstein spaces}
\label{sec:InterpolSeq}

\paragraph{Signal spaces}
Throughout this paper $\TT := [-T/2,T/2]$ stands for the closed interval of length $T$.
Since our signals will be sampled in the frequency domain, we define our signal spaces $\ourSpace^{p}$ in terms of their Fourier representations as
\begin{equation}
\label{equ:SignalSpace}
	\ourSpace^{p} := \{ x \in \mathcal{S}'\ :\ \Hx = \mathcal{F} x \in \ourB^{p} \}
	\quad\text{for any}\ 1 \leq p \leq \infty
\end{equation}
of all tempered distributions whose Fourier-Laplace transforms belong to $\ourB^{p}$.
We call $x$ and $\Hx = \mathcal{F} x$ the signal in the \emph{time domain} and \emph{Fourier domain} respectively.
By Theorem~\ref{thm:PaleyWiener}, every $x \in \ourSpace^{p}$ has a compact support contained in $\TT$.

\begin{remark}
Note that for $p>2$ the signal space $\ourSpace^{p}$ contains non-regular distributions.
Conversely for $p<2$, since $\L^{p}(\RN) \subset \mathcal{S}'$ for all $p$ we know that $\mathcal{F}:\L^{p'}(\RN)\rightarrow \L^{p}(\RN)$ is injective. Thus every $x \in \ourSpace^{p}$ is a function in $\L^{p'}(\TT)$. However $\mathcal{F}:\L^{p}(\RN)\rightarrow \L^{p'}(\RN)$ for $p\neq 2$ is not surjective and the same holds for $\mathcal{F}^{-1}:\L^{p}(\RN)\rightarrow \L^{p'}(\RN)$. Therefore $\ourSpace^{p}$ is not the entire space $\L^{p'}(\TT)$ since there are functions $x \in \L^{p'}(\TT)$ for which the restriction of $\mathcal{F} x$ on $\RN$ is not in $\L^{p}(\RN)$ but only in $\mathcal{S}'$.
\end{remark}

\paragraph{Interpolation in the Fourier and time domain}
A sequence $\Lambda = \{\lambda_{n}\}_{n\in\ZN}$ of distinct points in $\CN$ is said to be \emph{complete interpolating} for $\ourB^{p}$, $p\in (1,\infty)$ if for every sequence $\{c_{n}\}_{n\in\ZN} \in \ell^{p}$ there exists a unique $\Hx \in \ourB^{p}$ which solves the interpolation problem
\begin{equation}
\label{equ:interpolProb}
	\Hx(\lambda_{n}) = c_{n}
	\qquad\text{for all}\ n\in\ZN\;.
\end{equation}
In general, it is complicated to characterize complete interpolating sequences and to construct a function which fulfills condition \eqref{equ:interpolProb}.
However, for $\ourB^{p}$ with $1 < p < \infty$ a large class of complete interpolating sequences and the corresponding solution of \eqref{equ:interpolProb} is known (see, e.g., \cites{Young2001_Nonharmonic,Levin1997_Lectures}).
It consists of the zero sets of sine-type functions which are defined in the following.

\begin{definition}
An entire function $S$ of exponential type is said to be a \emph{sine-type function of type $\sigma$} if it has simple and separated zeros and if there exist positive constants $A,B,H$ such that
\begin{equation}
\label{equ:SineTypeH}
	A\, \E^{\sigma|\eta|} \leq |S(\xi + \I \eta)| \leq B\, \E^{\sigma|\eta|}
	\qquad\text{for}\ |\eta| \geq H\ \text{and all}\ \xi\in\RN\;.
\end{equation}
\end{definition}

Note also that sequences remain zeros of sine-type functions if one modifies their imaginary parts only. 
Furthermore, if $\Lambda = \{\lambda_{n}\}_{n\in\ZN}$ is the zero set of a sine-type function $S$ then one can factorize $S$ as follows
\begin{equation*}
	S(z)
	= z^{\delta_\Lambda} \lim_{R\to\infty} \prod_{\substack{|\lambda_{n}| < R \\ \lambda_{n} \neq 0}} \left( 1 - \frac{z}{\lambda_{n}} \right)
	\quad\text{with}\quad
	\delta_{\Lambda}
	= \left\{\begin{array}{ll}
	1 & \text{if}\ 0\in \Lambda\\
	0 & \text{otherwise}
	\end{array}
	\right.
\end{equation*}
where the infinite product converges uniformly on each compact set in $\CN$.

\begin{lemma}
\label{lem:ConvBp}
Let $\Lambda = \{\lambda_{n}\}_{n\in\ZN}$ be the zero set of a sine-type function $S$ of type $T/2$ and define
\begin{equation}
\label{equ:ReconstKernels}
	\widehat{\psi}_{n}(z)
	:= \frac{S(z)}{S'(\lambda_{n}) (z - \lambda_{n})}\;,
	\qquad z\in\CN\;.
\end{equation}
Then for any sequence $\bsc = \{c_{n}\}_{n\in\ZN} \in \ell^{p}$ with $1 \leq p < \infty$, the series
\begin{equation}
\label{equ:ReconstrInB}
	\Hx(z) = \sum_{n\in\ZN} c_n\, \widehat{\psi}_{n}(z)
\end{equation}
converges uniformly on each horizontal strip in the complex plane to a function in $\ourB^{p}$, which is the unique solution of \eqref{equ:interpolProb}.
For $1 < p < \infty$ the series also converges in the norm of $\ourB^{p}$ and there exist constants $0 < C_{1}\leq C_2$ such that
\begin{equation*}
	C_{1}\, \|\{ \Hx(\lambda_{n}) \}_{n\in\ZN}\|_{\ell^p}
	\leq \|\Hx\|_{\ourB^{p}}
	\leq C_{2}\, \|\{ \Hx(\lambda_{n}) \}_{n\in\ZN}\|_{\ell^p}
	\quad\text{for all}\ \Hx\in\ourB^{p}\;.
\end{equation*}
\end{lemma}

\begin{remark}
In particular $\widehat{\psi}_{n} \in \ourB^{p}$ with $\|\widehat{\psi}_{n}\|_{\ourB^{p}} \leq C_{2}$ for every $1 < p < \infty$.
\end{remark}

Lemma~\ref{lem:ConvBp} directly leads to the signal reconstruction formula for $\Hx\in\ourB^{p}$ in the Fourier domain.
By the linearity of the Fourier transform it can also be written in the time domain, which gives the following statement.

\begin{lemma}
\label{lem:ReconstrTD}
Let $\Lambda = \{\lambda_{n}\}_{n\in\ZN}$ be the zero set of a sine-type function $S$ of type $T/2$
and let $\psi_{n} = \mathcal{F}^{-1}\widehat{\psi}_{n}$ be the inverse Fourier transforms of the functions defined in \eqref{equ:ReconstKernels}. Then for each $1 \leq p < \infty$,
we have
\begin{equation}
\label{equ:TimeReconstruct}
	x = \sum_{n\in\ZN} \Hx(\lambda_{n})\, \psi_{n}
	\qquad\text{for all}\ x \in \ourSpace^{p}
\end{equation}
where the sum converges in the topology of $\mathcal{S}'$.
\end{lemma}

\begin{proof}
Let $\Hx_{N} := \sum^{N}_{n=-N} \Hx(\lambda_{n})\, \widehat{\psi}_{n}$, then  $\|\Hx - \Hx_{N}\|_{\ourB^{p}} \to 0$ as $N\to\infty$ for all $\Hx \in \ourB^{p}$
by Lemma~\ref{lem:ConvBp}.
Using \eqref{equ:boundedBp}, it follows that
\begin{eqnarray*}
	& \left| \int_{\RN} [\Hx(\omega) - \Hx_{n}(\omega) ]\,  \phi(\omega)\, \d\omega \right|
	\leq \|\Hx - \Hx_{n}\|_{\infty} \left| \int_{\RN} \phi(\omega)\, \d\omega \right|
	\leq C_{\phi}\, \|\Hx - \Hx_{n}\|_{\ourB^{p}}
\end{eqnarray*}
for all $\phi\in\mathcal{S}$ with a constant $C_{\phi}$ which depends only on $\phi$, $p$, and $T$. Consequently, $\Hx_{N} \to \Hx$ in $\mathcal{S}'$, and since
$\mathcal{F}^{-1}$ is a bijective and continuous mapping $\mathcal{S}' \to \mathcal{S}'$ it follows that $x_{N} \to x$ in $\mathcal{S}'$ for all $x \in \ourSpace^{p}$.
\end{proof}

\begin{remark}
If $x \in \ourSpace^{p}$ is a regular distribution, i.e. $x \in \L^{p'}(\TT)$, then \eqref{equ:boundedFTrafo} shows that for $1 < p \leq 2$ the series \eqref{equ:TimeReconstruct}
converges even in the $\L^{p'}$-norm to $x$.
However, for $p > 2$, the series \eqref{equ:TimeReconstruct} need not converge in the norm of $\L^{p'}$ even if $x \in \L^{p'}(\TT)$.
\end{remark}

Note that Lemma~\ref{lem:ConvBp} and \ref{lem:ReconstrTD} do not hold for $p=\infty$.
More specifically, the series $\sum_{n\in\ZN} \Hx(\lambda_n)\, \widehat{\psi}_{n}$ does not converge uniformly on $\RN$ for every $\Hx \in \ourB^{\infty}$.
However, a useful observation here is that in the particular case of $x \in \L^{1}(\TT)$, we have $\Hx \in C_{0}(\RN) \cap \ourB^{\infty}$.
For such functions the following lemma states that the series \eqref{equ:ReconstrInB} converges if oversampling is applied.

\begin{lemma}[\cite{MoenichBoche_SP2010}]
\label{lem:ConvBinfty}
Let $\Lambda = \{\lambda_{n}\}_{n\in\ZN}$ be the zero set of a sine-type function $S$ of type $T'/2$.
If $T < T'$ then for every $\Hx \in C_{0}(\RN) \cap \ourB^{\infty}$ we have
\begin{equation*}
	\lim_{N\to\infty} \max_{\omega \in \RN}
	\left| \Hx(\omega) - \sum^{N}_{n=-N} \Hx(\lambda_{n})\, \widehat{\psi}_{n}(\omega) \right|
	= 0\;.
\end{equation*}
\end{lemma}

\begin{remark}
In \cite{MoenichBoche_SP2010} only the case of real zeros $\lambda_{n}$ was treated.
However, the corresponding proof can also be applied for complex zeros, using that all zeros of a sine-type function lie in a strip parallel to the real axis.
\end{remark}

It now follows that Lemma~\ref{lem:ReconstrTD} may also be formulated for all regular distributions in $\ourSpace^{\infty}$. 
To this end, one has to modify the conditions slightly and suppose that $\Lambda$ is the zero set of a sine-type function of type $T'/2 > T/2$ in order to subsequently use Lemma~\ref{lem:ConvBinfty} in the corresponding proof.

\begin{example}
\label{exa:Shannon}
The zeros of the sine-type function $S(z) = \sin(\frac{T'}{2} z)$ of type $T'/2$ are given by $\lambda_{n} = n \tfrac{2\pi}{T'}$, $n\in\ZN$.
The corresponding functions \eqref{equ:ReconstKernels} are then given by $\widehat{\psi}_n(z) = \sinc(\frac{T'}{2}[z-n\frac{2\pi}{T'}])$ where $\sinc(x) := \sin(x)/x$.
Then the reconstruction series \eqref{equ:ReconstrInB} is known as Shannon sampling series, and it converges for all $\Hx \in \ourB^{p}$ with $1 < p < \infty$, provided that $T \leq T'$.
The corresponding time-domain series \eqref{equ:TimeReconstruct} becomes
$x(t) = \sum_{n\in\ZN} \Hx(\lambda_{n})\, \E^{-\I n \frac{2\pi}{T'} t}$ for all $t\in\TT$.
\end{example}

\section{Measurement Methodology}
\label{sec:Measurement}

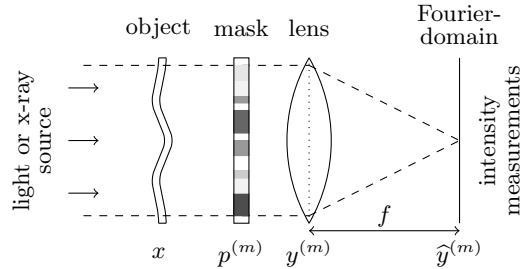
\begin{figure}[t]
\begin{center}
\begin{tikzpicture}
	\draw (-3.8,0) node[rotate = 90] {\footnotesize light or x-ray};
	\draw (-3.5,0) node[rotate = 90] {\footnotesize source};
	\draw[->] (-3.2,-0.7) -- (-2.8,-0.7);
	\draw[->] (-3.2,0) -- (-2.8,0);
	\draw[->] (-3.2,0.7) -- (-2.8,0.7);
	\draw[dashed] (-3,1) -- (0,1);
	\draw[dashed] (-3,-1) -- (0,-1);
	\draw[dashed] (0,1) -- (2,0);
	\draw[dashed] (0,-1) -- (2,0);
	\draw  (-2,1.5) node {\footnotesize object};
	\draw[rounded corners = 1mm] (-2, -1.1) -- (-2,-1) -- (-2.1,-0.5) -- (-1.9,0) -- (-2.1,0.5) -- (-2,1) -- (-2,1.1);
	\draw[rounded corners = 1mm] (-1.9, -1.1) -- (-1.9,-1) -- (-2.0,-0.5) -- (-1.8,0) -- (-2,0.5) -- (-1.9,1) -- (-1.9,1.1);
	\draw (-1.9,-1.1) -- (-2,-1.1);
	\draw (-1.9,1.1) -- (-2,1.1);
	\draw  (-2,-1.5) node {\footnotesize $x$};
	\draw  (-0.9,1.5) node {\footnotesize mask};
	\filldraw[fill = black!70!white, draw = black!70!white] (-1,-1) rectangle (-0.8,-0.7);
	\filldraw[fill = black!05!white, draw = black!05!white] (-1,-0.7) rectangle (-0.8,-0.4);
	\filldraw[fill = black!20!white, draw = black!20!white] (-1,-0.5) rectangle (-0.8,-0.4);
	\filldraw[fill = black!40!white, draw = black!40!white] (-1,-0.2) rectangle (-0.8,-0.0);
	\filldraw[fill = black!60!white, draw = black!60!white] (-1,0.1) rectangle (-0.8,0.4);
	\filldraw[fill = black!40!white, draw = black!40!white] (-1,0.5) rectangle (-0.8,0.6);
	\filldraw[fill = black!05!white, draw = black!05!white] (-1,0.6) rectangle (-0.8,0.8);
	\filldraw[fill = black!10!white, draw = black!10!white] (-1,0.8) rectangle (-0.8,1.0);
	\draw (-1,-1.1) rectangle (-0.8,1.1);
	\draw  (-0.9,-1.5) node {\footnotesize $p^{(m)}$};
	\draw  (0,1.5) node {\footnotesize lens};
	\draw (0,1.1) arc (150:210:2.2);
	\draw (0,-1.1) arc (-30:30:2.2);
	\draw[dotted] (0,-1.1) -- (0,1.1);
	\draw  (0,-1.5) node {\footnotesize $y^{(m)}$};
	\draw  (2,1.7) node {\footnotesize Fourier-};
	\draw  (2,1.4) node {\footnotesize domain};
	\draw (2,1.1) -- (2,-1.1);
	\draw  (2,-1.5) node {\footnotesize $\Hy^{(m)}$};
	\draw[<->] (0,-1.2) -- (2,-1.2);
	\draw (1,-1) node {\footnotesize $f$};
	\draw (2.4,0) node[rotate = 90] {\footnotesize intensity};
	\draw (2.7,0) node[rotate = 90] {\footnotesize measurements};
\end{tikzpicture}
\end{center}
\caption{Schematic setup for structured modulation in optics using spatial light modulators (masks).}
\label{fig:Optics}
\end{figure}

We consider signals in the signal space \eqref{equ:SignalSpace} and apply a measurement methodology which
applies several modulations of the desired signals with complex exponentials (structured modulations) before sampling.
Suppose $x \in \ourSpace^{p}$ is the signal of interest.
When phase loss is intrinsic to the measurement procedure, it is often still possible to modify the desired signal in different ways before the actual measurement.
In optical applications one may apply spatial light modulators for this purpose. An example of a corresponding measurement setup  (see also, e.g.,  \cites{Zhang_ApatureMod07,KatkovnikAstola_12,CandesEldar_PhaseRetrieval}) is sketched in Fig.~\ref{fig:Optics}.
There the object of interest is illuminated by a coherent light source which produces a diffraction pattern $x(t)$\footnote{The variable $t$ stands here for the one or two dimensional spatial dimension.}.
This diffraction pattern is modified by a spatial light modulator (a mask) with complex-valued spatial transmittance $p^{(m)}(t)$. 
In our setup we choose specific masks with transmittance functions $p^{(m)}$ with $m = 1, 2, \dots, M$. 
The resulting diffraction pattern then reads $y^{(m)}(t) = x(t)\, p^{(m)}(t)$.
This pattern is transformed into the frequency domain by the subsequent lens or diffraction imaging (see, e.g., \cite{Goodman_FourierOptik}).
On the measurement plane, the moduli of the frequency transforms $|\Hy^{(m)}(\omega)|^{2} = |(\mathcal{F}y^{(m)})(\omega)|^{2}$ are uniformly sampled at discrete points $\omega_{n} = n\beta$.
A more schematic illustration of the setup is provided in Fig.~\ref{fig:MeasureSetup}. 

We now turn to the choice of $p^{(m)}$. It will be shown that the general form
\begin{equation}
\label{equ:ModSig}
	p^{(m)}(t) = \sum^{K}_{k=1} \overline{\alpha^{(m)}_{k}} \E^{-\I\lambda_k t}
\end{equation}
allows for signal recovery. 
Then the samples in the $m$th branch are given by
\begin{equation}
\label{equ:IntensMeasure}
	c^{(m)}_{n}
	= |\Hy^{(m)}(n \beta)|^{2}
	= \left| \sum^{K}_{k=1} \overline{\alpha^{(m)}_{k}}\,  \Hx(n \beta + \lambda_{k}) \right|^{2} 
	= \left| \langle  \widehat{\bx}_{n},\balpha_m \rangle_{\CN^{K}} \right|^2\;,
	\quad n\in\ZN
\end{equation}
with the length $K$ vectors
\begin{equation}
\label{equ:alpha_xn}
\balpha_m := \left(\begin{array}{c}
	\alpha_1^{(m)} \\ \vdots \\ \alpha_K^{(m)}
	\end{array}\right)
	\qquad \text{and}\qquad
\Hbx_{n} := \left(\begin{array}{c}
	\Hx(n\beta + \lambda_1) \\ \vdots \\ \Hx(n\beta + \lambda_K)
	\end{array}\right)\;.
\end{equation}
One of the main results imply that if the vectors $\balpha_m$, $m=1,\dots,M$ and the interpolation points $\lambda_{n,k} := n\beta + \lambda_{k} : n\in\ZN\;, k=1,\dots,K$ are properly chosen, then it is possible to reconstruct $x$ from all samples $\bc = \{c^{(m)}_{n} : m=1,\dots,M\,; n\in\ZN\}$. 
The reconstruction procedure consists of three steps.
First for each $n\in \ZN$ one determines the vector $\Hbx_{n}$ from the $M$ measurements~\eqref{equ:IntensMeasure} up to a unimodular factor.
Then, one determines and matches the unimodular factors in the sequence of vectors $\{\Hbx_{n}\}$ and finally reconstructs the continuous $x$ from the entries of the vectors $\Hbx_{n}$ using Lemma~\ref{lem:ReconstrTD}.

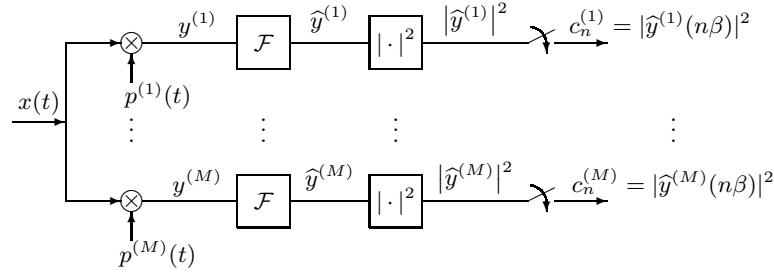
\begin{figure}[t]
\begin{center}
\begin{picture}(290,95)(10,25)
	\put(20,77){\makebox(0,0){\footnotesize $x(t)$}}
	\put(10,70){\vector(1,0){20}}
	\put(30,40){\line(0,1){60}}
	\put(30,100){\vector(1,0){21}}
	\put(55,100){\circle{8}}		\put(55,100){\makebox(0,0){\footnotesize$\times$}}
	\put(55,85){\vector(0,1){11}}	\put(65,80){\makebox(0,0){\footnotesize $p^{(1)}(t)$}}
	\put(59,100){\line(1,0){36}}
	\put(80,107){\makebox(0,0){\footnotesize $y^{(1)}$}}
	\put(55,70){\makebox(0,0){\footnotesize $\vdots$}}
	\put(30,40){\vector(1,0){21}}
	\put(55,40){\circle{8}}   	\put(55,40){\makebox(0,0){\footnotesize$\times$}}
	\put(55,25){\vector(0,1){11}}	\put(65,20){\makebox(0,0){\footnotesize $p^{(M)}(t)$}}
	\put(59,40){\line(1,0){36}}
	\put(80,47){\makebox(0,0){\footnotesize $y^{(M)}$}}
	\put(95,110){\line(1,0){20}}
	\put(95,90){\line(1,0){20}}
	\put(95,90){\line(0,1){20}}
	\put(115,90){\line(0,1){20}}
	\put(105,100){\makebox(0,0){\footnotesize $\mathcal{F}$}}
	\put(115,100){\line(1,0){30}}
	\put(130,109){\makebox(0,0){\footnotesize $\Hy^{(1)}$}}
	\put(145,110){\line(1,0){20}}
	\put(145,90){\line(1,0){20}}
	\put(145,90){\line(0,1){20}}
	\put(165,90){\line(0,1){20}}
	\put(155,100){\makebox(0,0){\footnotesize $\left|\, \cdot\, \right|^{2}$}}
	\put(165,100){\line(1,0){40}}	
	\put(185,109){\makebox(0,0){\footnotesize $\big|\Hy^{(1)}\big|^{2}$}}
	\put(105,70){\makebox(0,0){\footnotesize $\vdots$}}
	\put(155,70){\makebox(0,0){\footnotesize $\vdots$}}	
	\put(95,50){\line(1,0){20}}
	\put(95,30){\line(1,0){20}}
	\put(95,30){\line(0,1){20}}
	\put(115,30){\line(0,1){20}}
	\put(105,40){\makebox(0,0){\footnotesize $\mathcal{F}$}}
	\put(115,40){\line(1,0){30}}
	\put(130,49){\makebox(0,0){\footnotesize $\Hy^{(M)}$}}
     \put(145,50){\line(1,0){20}}
	\put(145,30){\line(1,0){20}}
	\put(145,30){\line(0,1){20}}
	\put(165,30){\line(0,1){20}}
	\put(155,40){\makebox(0,0){\footnotesize $\left|\, \cdot\, \right|^{2}$}}
	\put(165,40){\line(1,0){40}}
	\put(185,49){\makebox(0,0){\footnotesize $\big|\Hy^{(M)}\big|^{2}$}}
	\put(205,100){\line(2,1){10}}
	\qbezier(206,106)(211,105)(212,100)	\put(212,98){\vector(0,-1){2}}
	\put(215,100){\vector(1,0){20}}
	\put(255,107){\makebox(0,0){\footnotesize $c^{(1)}_{n} = |\Hy^{(1)}(n\beta)|^{2}$}}
	\put(260,70){\makebox(0,0){\footnotesize $\vdots$}}		
	\put(205,40){\line(2,1){10}}
	\qbezier(206,46)(211,45)(212,40)		\put(212,38){\vector(0,-1){2}}
	\put(215,40){\vector(1,0){20}}
	\put(260,47){\makebox(0,0){\footnotesize $c^{(M)}_{n}  = |\Hy^{(M)}(n\beta)|^{2}$}}
\end{picture}
\end{center}
\caption{Measurement setup: In each branch, the unknown signal $x$ is modulated with a different sequence $p^{(m)}$, $m=1,2,\dots,M$.
Subsequently, the intensities of the resulting signals $y^{(m)}$ are measured and uniformly sampled in the frequency domain.}
\label{fig:MeasureSetup}
\end{figure}

\subsection{Choice of the modulation coefficients $\alpha^{(m)}_{k}$} 
\label{sec:Balan}
In order to determine the vector $\Hbx_{n} \in \CN^{K}$ from the $M$ intensity measurements \eqref{equ:IntensMeasure}, we apply a result from \cite{Balan_Painless}.
It states that if the family of $\CN^{K}$-vectors $\mathcal{A} = \{\balpha_1, \dots,\balpha_M \}$ constitutes a $2$-uniform $M/K$-tight frame which contains $M = K^{2}$ vectors or if $\mathcal{A}$ is a union of $K+1$ mutually unbiased bases in $\CN^{K}$, then every $\Hbx_{n} \in \CN^{K}$ can be reconstructed up to a constant phase from the magnitude of the inner products \eqref{equ:IntensMeasure}.
We only discuss the first case here and therefore fix $M = K^{2}$. The adaption to the second case is obvious.
\begin{condA}
A set $\{\alpha^{(m)}_{k}\}$ of coefficients in \eqref{equ:ModSig} is said to satisfy \emph{Condition~A} if the set $\{\balpha_1, \dots,\balpha_M \}$ of vectors in $\CN^{K}$, as defined in \eqref{equ:alpha_xn}, constitutes a $2$-uniform $M/K$-tight frame with $M = K^2$ vectors.
\end{condA}

Reconstruction will then be based on the following result:
\begin{theorem}[\cites{Balan_Painless,Levenshtein_98}]
Let $\{\balpha_1, \dots,\balpha_M \}$ be a uniform $M/K$-tight frame in $\CN^{K}$ with $M = K^{2}$ vectors.
Then for every Hermitian rank-$1$ matrix $Q_{n} = \Hbx_{n}\Hbx^{*}_{n}$ we have
\begin{equation}
\label{equ:BalanReconstruct}
Q_{n} 
= \frac{K+1}{K}\sum_{m=1}^{M} \left| \langle  \widehat{\bx}_{n},\balpha_m \rangle_{\CN^{K}} \right|^2
  \, \balpha_m\balpha^{*}_m - \frac{1}{K}\sum_{m=1}^{M}c_n^{(m)}I.
\end{equation}
\end{theorem}

Constructions of 2-uniform $M/K$-tight frames with $M=K^2$ vectors and for different dimensions $K$ can be found in \cite{Zauner_Quantendesigns}.
For $K=2$ such a frame is given by \cite{Balan_Painless}
\begin{equation*}
	\balpha_1 =  \binom{a}{b},\
	\balpha_2 =  \binom{b}{a},\
	\balpha_3 =  \binom{a}{-b},\
	\balpha_4 =  \binom{-b}{a}
\end{equation*}
with $a = \sqrt{\frac{1}{2}(1-\frac{1}{\sqrt{3}})}$ and $b = \E^{\I 5\pi/4}\sqrt{\frac{1}{2}(1+\frac{1}{\sqrt{3}})}$.

\subsection{Choice of the interpolation points $\lambda_{k}$}
\label{sec:Sampling}

Let $\{\lambda_k\}_{k=1}^K$ be ordered increasingly by their real parts.
For each $n\in\ZN$, the vector $\Hbx_{n}$ in \eqref{equ:alpha_xn} contains the values of $\Hx$ at $K$ distinct interpolation points in the complex plane collected in the sequences
\begin{equation}
\label{equ:SampSets}
	\blambda^a_n := \{\lambda^a_{n,k}\}^{K}_{k = 1}
	\quad\text{with}\quad
	\lambda^a_{n,k} =  n\beta + \lambda_k\;,
	\quad n\in\ZN\;.
\end{equation}
Therein, the parameter $a\in\NN$ denotes the number of overlapping points of consecutive sets \eqref{equ:SampSets} (cf. Fig.~\ref{fig:SpampPoints}).
More precisely, we require for every $n\in\ZN$ that
\begin{equation}
\label{equ:CondSampPoint}
\lambda^a_{n,i} = \lambda^a_{n-1,K-i+1}\quad \text{for all}\ i = 1,\dots, a\;.
\end{equation}
In the following $\Lambda_{O,n}^a = \blambda^a_n \cap \blambda^a_{n+1}$ denotes the set of overlapping interpolation points between $\blambda^a_n$ and $\blambda^a_{n+1}$,
and we define the overall interpolation sequence
\begin{eqnarray}
\label{equ:Lambda_a}
& \Lambda^a := \bigcup_{n\in\ZN}\blambda^a_n
   = \bigcup_{n\in\ZN} \{n\beta + \lambda_{k}\}^{K}_{k=1}\;.
\end{eqnarray}
In general we allow for $a\geq 1$, but we will see that $a=1$ is generally sufficient for reconstruction.

As explained in Section~\ref{sec:InterpolSeq}, $x\in\ourSpace^{p}$ can be perfectly reconstructed by \eqref{equ:TimeReconstruct} if $\Lambda^a$ is complete interpolating for $\ourB^{p}$.
Consequently, we require the following condition to hold for the interpolation points $\{\lambda_{k}\}$.

\begin{figure}[t]
\begin{center}
\begin{tikzpicture}
	\draw[->] (-3.9,0) -- (3.9,0);
	\draw (3.8,-0.3) node {$\xi$};
	\draw[->] (0,-1.0) -- (0,1.7);
	\draw (0.2,1.6) node {$\eta$};
	\draw[dotted] (-2,-0.9) -- (-2,1.4);
	\draw[dotted] (2,-0.9) -- (2,1.0);
	\draw[<->,dashed] (-2,-0.7) -- (0,-0.7);
	\draw (-1,-0.9) node {$\beta$};
	\draw[black, thick] (-3.5,0.0) circle (0.05);
	\draw[black, thick] (-3.0,0.3) circle (0.05);
	\draw[black, thick] (-2.5,-0.4) circle (0.05);
	\draw[black, thick] (-2.0,0.6) circle (0.05);
	\draw[black, thick] (-1.5,0.0) circle (0.05);		
	\draw[black, thick] (-2.0,0.6) circle (0.05);
	\draw[black, thick] (-1.5,0.0) circle (0.05);
	\draw[black, thick] (-1.0,0.3) circle (0.05);
	\draw[black, thick] (-0.5,-0.4) circle (0.05);
	\draw[black, thick] (0.0,0.6) circle (0.05);
	\draw[black, thick] (0.5,0.0) circle (0.05);
	\filldraw[fill = black!20!white, fill opacity=0.3, draw = black, dotted, rounded corners = 1mm]
	           (-2.15,0.6) -- (-2.0,0.8) -- (-1.5,0.3) -- (-1.0,0.5) -- (-0.5,0.5) -- (0.0,0.8) -- (0.5,0.2) -- (0.7,0.0) --
	           (0.5,-0.2) -- (0.0,-0.1) -- (-0.5,-0.6) -- (-1.0,-0.2) -- (-1.5,-0.2) -- (-2.0,0.2) -- (-2.15,0.6);
	\draw[black] (-0.45,0.21) node{$\blambda^{a}_{n-1}$};
	\filldraw[black, thick] (0.0,0.6) circle (0.05);
	\filldraw[black, thick] (0.5,0.0) circle (0.05);
	\draw[black, thick] (1.0,0.3) circle (0.05);
	\draw[black, thick] (1.5,-0.4) circle (0.05);
	\draw[black, thick] (2.0,0.6) circle (0.05);
	\draw[black, thick] (2.5,0.0) circle (0.05);
	\filldraw[fill = black!05!white, fill opacity=0.3, draw = black, dashed, rounded corners = 1mm]
	           (-0.15,0.6) -- (0.0,0.8) -- (0.5,0.3) -- (1.0,0.5) -- (1.5,0.5) -- (2.0,0.8) -- (2.5,0.2) -- (2.7,0.0) --
	           (2.5,-0.2) -- (2.0,-0.1) -- (1.5,-0.6) -- (1.0,-0.2) -- (0.5,-0.2) -- (0.0,0.2) -- (-0.15,0.6);
	\draw[black] (1.5,0.21) node{$\blambda^{a}_{n}$};
	\draw[black, thick] (2.0,0.6) circle (0.05);
	\draw[black, thick] (2.5,0.0) circle (0.05);
	\draw[black, thick] (3.0,0.3) circle (0.05);
	\draw[black, thick] (3.5,-0.4) circle (0.05);
	\draw[->] (1.2,1.2) -- (0.1,0.7);
	\draw[black] (2.2,1.3) node{\footnotesize $\lambda^{a}_{n-1,5} = \lambda^{a}_{n,1}$};
\end{tikzpicture}
\end{center}
\caption{Illustration for the choice of interpolation points in the complex plane for $K=6$ in \eqref{equ:ModSig} and an overlap $a = 2$.}
\label{fig:SpampPoints}
\end{figure}
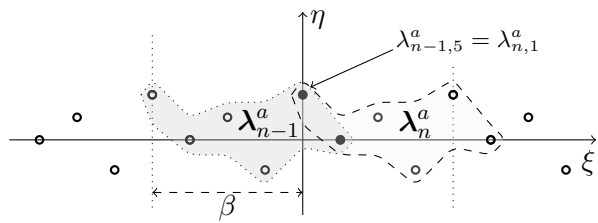

\begin{condB}
A the set $\{\lambda_{k}\}^{K}_{k=1}$ of coefficients in \eqref{equ:ModSig} is said to satisfy Condition~B if there exists a $\beta > 0$ such that
the set $\Lambda^{a}$, defined in \eqref{equ:Lambda_a}, 
is complete interpolating for $\ourB^{p}$ and satisfies \eqref{equ:CondSampPoint} for a certain $1 \leq a < K$.
\end{condB}

If $\{\lambda_{k}\}^{K}_{k=1}$ satisfies Condition~B, then the overall interpolation sequence $\Lambda^{a}$ in \eqref{equ:Lambda_a} is $\beta$-periodic.
One specific way to obtain such sequences is to choose $\Lambda^a$ as the set of zeros of a $\beta$-periodic sine-type function of type $\widetilde{T}/2 \geq T/2$. 
Based on such a zero set, one may modify the imaginary part of the individual interpolation points or one can perturb the individual interpolation points slightly (according to Katsnelson's theorem) without changing the complete interpolating property and such that $\beta$-periodicity is preserved.
This way, it is possible to construct many different sets $\{\lambda_{k}\}^{K}_{k=1}$ which satisfy Condition~B. One particularly simple construction is obtained by starting with the zeros of the sine-type function $\sin( \widetilde{T}/2\, z )$, which has equally spaced zeros on the real axis (cf. Example~\ref{exa:Shannon}) and consequently $\beta$ being an integer multiple of $4\pi/\widetilde{T}$.

\section{Phaseless Signal Recovery}
\label{sec:MainResult}

We assume a sampling scheme as described in Section~\ref{sec:Measurement} (cf.~Fig.~\ref{fig:MeasureSetup})
with modulating functions $p^{(m)}$ of the form \eqref{equ:ModSig} which satisfy Condition~A and B.
For this setup, we show that every signal in a countable intersection of open dense subsets of $\ourSpace^{p}$ can be reconstructed from the samples in \eqref{equ:IntensMeasure}.
The proof provides an explicit algorithm for signal recovery. In principle, it consists of a three step procedure.
First, a finite block of $K$ samples of the Fourier domain signal $\Hx$ is reconstructed up to a constant phase factor using the $m=1,\dots,M$ intensity measurements \eqref{equ:IntensMeasure} taken at sampling instants $n$. Since Condition A is fulfilled, we can apply a finite dimensional phase retrieval algorithm to achieve that. 
In the second step we exploit that by our construction of the interpolation points, consecutive blocks have an overlap. Therewith, it is possible to make the unimodular factors consistent over all blocks.
Finally Lemma~\ref{lem:ReconstrTD} is applied to interpolate the continuous signal, utilizing Condition~B of the sampling scheme.

\begin{theorem}
\label{thm:MainThm}
For any $K\geq 2$, let $\{\alpha^{(m)}_{k}\}$ be a set of modulation coefficients which satisfy Condition~A,
and let $\{\lambda_{k}\}^{K}_{k=1}$ be a set which satisfies Condition~B.
For $1 < p < \infty$, let $x \in \ourSpace^{p}$.
If for each $n\in\ZN$ the set $\Hx(\Lambda^a_{O,n})$ contains at least one non-zero element, then $x$ can be perfectly reconstructed from the intensity measurements
\begin{equation}
\label{equ:IntensMeasure_Thm}
	c^{(m)}_{n}
	= \left| \sum^{K}_{k=1} \overline{\alpha^{(m)}_{k}}\,  \Hx(n \beta + \lambda_{k}) \right|^{2}\;, 
	\quad
	\begin{array}{l}
	   m=1,\dots,K^{2}\\
	   n\in\ZN
	\end{array}
\end{equation}
up to a constant phase.
\end{theorem}

\begin{proof}
Since $\{\lambda_{k}\}^{K}_{k=1}$ satisfies Condition~B, there exists a $\beta > 0$ such that $\Lambda^{a}$ is complete interpolating for $\ourB^{p}$. Therefore the signal $x$ can be reconstructed from the entries of the vectors $\{\Hbx_{n}\}_{n\in\ZN}$, defined in \eqref{equ:alpha_xn}, using \eqref{equ:TimeReconstruct}.
It remains to show that $\{\Hbx_{n}\}_{n\in\ZN}$ can be determined from measurements \eqref{equ:IntensMeasure_Thm}.

Let $n\in\ZN$ be arbitrary. Since $\{\alpha^{(m)}_{k}\}$ satisfies Condition~A, we can directly use the relation \eqref{equ:BalanReconstruct} to determine the rank-$1$ matrix $Q_{n} := \Hbx_n\Hbx_n^*$ from the measurements $\{c^{(m)}_{n}\}^{M}_{m=1}$. Then $\Hbx_n \in \CN^K$ is obtained by factorizing $Q_{n}$. However, such a factorization is only unique up to a constant phase factor.
If the phase $\phi_{n,i}$ of one element $[\Hbx_n]_i$ is known, the vector $\Hbx_n$ can be completely determined from $Q_{n}$ by
\begin{equation}
\label{equ:factorisation}
\Hx(n\beta +\lambda_k) = \sqrt{[Q_{n}]_{k,k}}\, \E^{\I(\phi_{n,i} - \arg([Q_{n}]_{i,k}))}, \ \forall k \neq i\,.
\end{equation}
Let the recovery of the sequence $\{\Hbx_{n}\}_{n\in\ZN}$ now start at an arbitrary $n_{0}\in\ZN$.
In this initial step, we set the constant phase of $\Hbx_{n_{0}}$ arbitrarily to $\theta_{0} \in [-\pi,\pi]$.
In the next step, we determine $\Hbx_{n_{0}+1}$. After the factorization of $Q_{n_{0}+1}$, the vector $\Hbx_{n_{0}+1}$ is only determined up to a constant phase. 
However, since $\Lambda^{a}_{O,n_0}$ is non-empty, and because $\Hx(\Lambda^a_{O,n_{0}})$ contains at least one non-zero element, we have phase knowledge of at least one entry of $\Hbx_{n_{0}+1}$, say $\Hx(\lambda^a_{n_{0}+1,i})$, where $\lambda^a_{n_{0}+1,i}$ is an overlapping interpolation point of $\blambda^a_{n_{0}}$ and $\blambda^a_{n_{0}+1}$.
Thus, we can completely determine $\Hbx_{n_{0}+1}$ and successively all  $n = n_{0}\pm 1, n_{0}\pm 2, \dots$ using \eqref{equ:factorisation} to obtain $\Hx(\Lambda^a)\, \E^{\I \theta_0}$.

The arbitrary phase $\theta_{0}$ of the initial vector $\Hbx_{n_0}$ yields a global constant phase shift $\theta_0$ for all $\Hbx_{n}$ which persists after the reconstruction of the time signal by the interpolation formula \eqref{equ:TimeReconstruct}.
\end{proof}

\begin{remark}
In the case $p=1$ and $p=\infty$, for regular distributions $x \in \L^1(\TT)$ and correspondingly $\Hx \in C_0 \cap \ourB^{\infty}$ or $\Hx \in C_0 \cap \ourB^1$, reconstruction is also possible according to Lemma \ref{lem:ConvBinfty}. Then however, Condition B has to be reformulated such that the interpolation points $\Lambda^{a}$ are the zero set of a sine-type function of type $T'/2>T/2$.
\end{remark}

Theorem~\ref{thm:MainThm} states that $x\in\ourSpace^{p}$ can only be reconstructed if $\Hx = \mathcal{F} x \in\ourB^{p}$ has at most $a-1$ zeros on the overlapping interpolation sets $\Lambda^{a}_{O,n}$.
Thus, the set $\mathcal{G}$ of all $\Hx \in \ourB^{p}$ which have $a$ zeros in at least one set $\Lambda^{a}_{O,n}$ contains all those functions for which the reconstruction procedure of Theorem~\ref{thm:MainThm} will fail. We now show that $\mathcal{G}$ is small in a certain sense. 

\begin{lemma}
The set $\mathcal{G}$ of all $\Hx \in \ourB^{p}$ for which the reconstruction procedure of Theorem~\ref{thm:MainThm} fails is of first category.
\end{lemma}

\begin{proof}
Let $\Lambda^{a} = \{\lambda_{n}\}_{n\in\ZN}$ be a set of interpolation points as applied in Theorem~\ref{thm:MainThm} and set $\mathcal{H}_{n} := \{ \Hx \in \ourB^{p}\ :\ \Hx(\lambda_{n}) \neq 0 \}$
for every $n\in\ZN$.
It follows from Plancherel-P{\'o}lya~\eqref{equ:boundedBp} that the point evaluations $\Hx \mapsto \Hx(z)$ are continuous on $\ourB^{p}$, and so $\mathcal{H}_n$ is open. 
Moreover, $\mathcal{H}_n$ is dense in $\ourB^p$ because to every $\Hy \in \ourB^{p}$ with $\Hy(\lambda_{n}) = 0$ and for every $\epsilon > 0$ we can find $\Hx = \Hy + (\epsilon/C_{2})\, \widehat{\psi}_{n} \in \mathcal{H}_{n}$,  where $\widehat{\psi}_{n}$ and $C_{2}$ are given in Lemma~\ref{lem:ConvBp},
and such that $\|\Hx - \Hy\| = (\epsilon/C_{2})\, \|\widehat{\psi}_{n}\| \leq \epsilon$.
Hence, the complements $\mathcal{G}_n := \mathcal{H}_n^C$ and $\bigcup_{n\in\ZN}\mathcal{G}_n$ are nowhere dense. Since $\mathcal{G} \subset \bigcup_{n\in\ZN} \mathcal{G}_{n}$, $\mathcal{G}$ is of first category.
\end{proof}

\section{Signal Reconstruction in Subspaces}
\label{sec:SubSpaces}

The restriction on the signal space given in Theorem~\ref{thm:MainThm} is fairly mild.
However, in order to avoid even such pathological cases, we may a priori minimally restrict the function space allowed in Theorem~\ref{thm:MainThm} to prevent the measured signal $\Hx$ from having zeros in $\Lambda^a$.
From the practical point of view, this can be achieved by adding a known test signal $u$ to the desired signal $x \in \ourSpace^{p}$ prior to the structured modulations, intensity measurement and sampling (cf. Fig.\ref{fig:MeasureSetup2}).

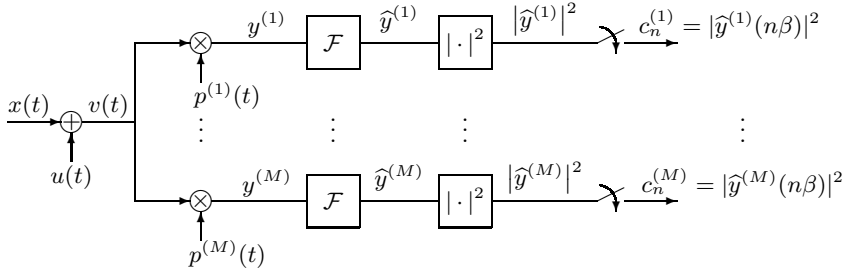
\begin{figure}[t]
\begin{center}
\begin{picture}(280,95)(10,25)
	\put(-18,70){\vector(1,0){20}}
	\put(-10,76){\makebox(0,0){\footnotesize $x(t)$}}
	\put(6,70){\circle{8}}		\put(6,70){\makebox(0,0){\footnotesize$+$}}
	\put(6,55){\vector(0,1){11}}	\put(6,49){\makebox(0,0){\footnotesize $u(t)$}}
	\put(20,76){\makebox(0,0){\footnotesize $v(t)$}}
	\put(10,70){\line(1,0){20}}
	\put(30,40){\line(0,1){60}}
	\put(30,100){\vector(1,0){21}}
	\put(55,100){\circle{8}}		\put(55,100){\makebox(0,0){\footnotesize$\times$}}
	\put(55,85){\vector(0,1){11}}	\put(65,80){\makebox(0,0){\footnotesize $p^{(1)}(t)$}}
	\put(59,100){\line(1,0){36}}
	\put(80,107){\makebox(0,0){\footnotesize $y^{(1)}$}}
	\put(55,70){\makebox(0,0){\footnotesize $\vdots$}}
	\put(30,40){\vector(1,0){21}}
	\put(55,40){\circle{8}}   	\put(55,40){\makebox(0,0){\footnotesize$\times$}}
	\put(55,25){\vector(0,1){11}}	\put(65,20){\makebox(0,0){\footnotesize $p^{(M)}(t)$}}
	\put(59,40){\line(1,0){36}}
	\put(80,47){\makebox(0,0){\footnotesize $y^{(M)}$}}
	\put(95,110){\line(1,0){20}}
	\put(95,90){\line(1,0){20}}
	\put(95,90){\line(0,1){20}}
	\put(115,90){\line(0,1){20}}
	\put(105,100){\makebox(0,0){\footnotesize $\mathcal{F}$}}
	\put(115,100){\line(1,0){30}}
	\put(130,109){\makebox(0,0){\footnotesize $\Hy^{(1)}$}}
	\put(145,110){\line(1,0){20}}
	\put(145,90){\line(1,0){20}}
	\put(145,90){\line(0,1){20}}
	\put(165,90){\line(0,1){20}}
	\put(155,100){\makebox(0,0){\footnotesize $\left|\, \cdot\, \right|^{2}$}}
	\put(165,100){\line(1,0){40}}	
	\put(185,109){\makebox(0,0){\footnotesize $\big|\Hy^{(1)}\big|^{2}$}}
	\put(105,70){\makebox(0,0){\footnotesize $\vdots$}}
	\put(155,70){\makebox(0,0){\footnotesize $\vdots$}}	
	\put(95,50){\line(1,0){20}}
	\put(95,30){\line(1,0){20}}
	\put(95,30){\line(0,1){20}}
	\put(115,30){\line(0,1){20}}
	\put(105,40){\makebox(0,0){\footnotesize $\mathcal{F}$}}
	\put(115,40){\line(1,0){30}}
	\put(130,49){\makebox(0,0){\footnotesize $\Hy^{(M)}$}}
     \put(145,50){\line(1,0){20}}
	\put(145,30){\line(1,0){20}}
	\put(145,30){\line(0,1){20}}
	\put(165,30){\line(0,1){20}}
	\put(155,40){\makebox(0,0){\footnotesize $\left|\, \cdot\, \right|^{2}$}}
	\put(165,40){\line(1,0){40}}
	\put(185,49){\makebox(0,0){\footnotesize $\big|\Hy^{(M)}\big|^{2}$}}
	\put(205,100){\line(2,1){10}}
	\qbezier(206,106)(211,105)(212,100)	\put(212,98){\vector(0,-1){2}}
	\put(215,100){\vector(1,0){20}}
	\put(255,107){\makebox(0,0){\footnotesize $c^{(1)}_{n} = |\Hy^{(1)}(n\beta)|^{2}$}}
	\put(260,70){\makebox(0,0){\footnotesize $\vdots$}}		
	\put(205,40){\line(2,1){10}}
	\qbezier(206,46)(211,45)(212,40)		\put(212,38){\vector(0,-1){2}}
	\put(215,40){\vector(1,0){20}}
	\put(260,47){\makebox(0,0){\footnotesize $c^{(M)}_{n}  = |\Hy^{(M)}(n\beta)|^{2}$}}
\end{picture}
\end{center}
\caption{Measurement setup as in Fig.~\ref{fig:MeasureSetup} but with an additional preprocessing of adding a test function $u(t)$ prior to modulation, Fourier transform, intensity measurement and sampling.}
\label{fig:MeasureSetup2}
\end{figure}

The first lemma is similar to a result by Duffin, Schaeffer \cite{DuffinSchaeffer_38}. It states that if one adds an additive test signal $u$ of which the Fourier-Laplace transform is a sine-type function of type $T'/2 > T/2$, then all zeros of the sum signal $v$ are located inside a strip parallel to $\RN$. In fact, $u$ can be chosen independent from the actual signal $x \in \ourSpace^{p}$.

\begin{lemma}
\label{lem:DufSchaef_2}
Let $\Hu$ be a sine-type function of type $T'/2$ and let $T < T'$ and $1 \leq p \leq \infty$ be arbitrary.
For $\Hx \in \ourB^{p}$ define the function $\Hv(z) := \Hx(z) + D\, \Hu(z)$.
Then $\Hv \in \mathcal{B}^{\infty}_{T'/2}$ and for every $D>0$ there exists an $H>0$ such that for all $\Hx \in \ourB^{p}$
\begin{equation}
\label{equ:DufSchaef}
	\big| \Hv(\xi + \I\eta) \big| > 0
	\quad\text{for all}\ |\eta| \geq H\ \text{and all}\ \xi\in\RN\;.
\end{equation}	
\end{lemma}

\begin{proof}
Since every sine-type function is bounded on $\RN$ and $\ourB^{p} \subset \mathcal{B}^{p}_{T'/2} \subset \mathcal{B}^{\infty}_{T'/2}$, it is immediately clear that $\Hv \in \mathcal{B}^{\infty}_{T'/2}$. Furthermore, for all $z\in\CN$ the triangle inequality implies
\begin{equation}
\label{equ:ProofDufSchaef_2}
	\big| \Hv(z) \big|
	\geq \big|\, | D\, \Hu(z) | - | \Hx(z) |\, \big|\;.
\end{equation}
It follows again from Plancherel-P{\'o}lya \eqref{equ:boundedBp} that there is a constant $M = C\, \|\Hx\|$ such that
\begin{equation*}
	\big| \Hx(\xi + \I\eta) \big|
	\leq M\, \E^{\frac{T}{2} |\eta|}
	\qquad\text{for all}\ \xi\in\RN\;.
\end{equation*}
Similarly, by the definition of a sine-type function in \eqref{equ:SineTypeH} one has the lower bound 
\begin{equation*}
	\big| D\, \Hu(\xi + \I\eta) \big| \geq D\, A_{u}\, \E^{\frac{T'}{2} |\eta|}
	\qquad\text{for all}\ |\eta| > H_{u}\
	\text{and all}\ \xi\in\RN\;,
\end{equation*}
with constants $A_{u}$ and $H_{u}$ which depend on $u$.
Using these two bounds in \eqref{equ:ProofDufSchaef_2} and keeping in mind that $T' > T$, one obtains
\begin{equation*}
	\big| \Hv(\xi + \I\eta) \big| \geq D\, A_{u}\, \E^{\frac{T'}{2} |\eta|} -  M\, \E^{\frac{T}{2} |\eta|} > 0
	\quad\text{for all}\ |\eta|>H\
	\text{and all}\ \xi\in\RN
\end{equation*}
with $H=\max\left[H_u,\frac{2}{T'-T}\ln\Big(\tfrac{M}{D\, A_{u}}\Big)\right]$.
\end{proof}

\begin{remark}
Conversely for a given $H > H_u$, one can choose the amplitude $D$ as
$D > \tfrac{M}{A_{u}}\, \exp\big([T-T']\,\tfrac{H}{2}\big)$
for \eqref{equ:DufSchaef} to hold.
\end{remark}

\begin{example}
\label{exa:Cosine}
One special choice for the function $\Hu$ that has been dealt with in \cite{DuffinSchaeffer_38}, is $\Hu(z) = \cos(\tfrac{T'}{2} z)$.
For this function, the constants~$A_{u}$ and $H_{u}$ can be derived by
\begin{align*}
        \left| \cos\left(\tfrac{T'}{2} z\right) \right|
	&= \frac{1}{2}\, \left| \E^{\I \tfrac{T'}{2} \xi}\, \E^{- \tfrac{T'}{2} \eta} + \E^{-\I \tfrac{T'}{2} \xi}\, \E^{\tfrac{T'}{2} \eta}\right|\\
	&\geq \frac{1}{2}\, \Big( \E^{\tfrac{T'}{2} |\eta|} - \E^{-\tfrac{T'}{2} |\eta|} \Big)
	\geq \frac{1}{2}\, \Big( \E^{\tfrac{T'}{2} |\eta|} - 1 \Big)\\
        &\geq \frac{1}{2}\, \Big(1-\E^{-\tfrac{T'}{2} H_{u}}\Big) \E^{\tfrac{T'}{2} |\eta|} =: A_{u}\, \E^{\tfrac{T'}{2} |\eta|}\;
	\quad\text{for all}\ |\eta| \geq H_{u}
\end{align*}
where $H_{u} > 0$ is arbitrary and $A_{u} = [1-\exp(- \tfrac{T'}{2} H_{u})]/2$.
Note that the time domain signal corresponding to this sine-type function $\Hu$ is the non-regular tempered distribution $u = [\delta_{T'/2} + \delta_{-T'/2}]/2$ which vanishes on $\TT$. 
\end{example}

By adding a sine-type function $\Hu$ of type $T'/2$ to the desired signal $\Hx \in \ourB^{p}$, one obtains a function $\Hv \in \mathcal{B}^{\infty}_{T'/2} \setminus C_0(\RN)$.
For these signals, the reconstruction formulas in Lemma~\ref{lem:ConvBp} and \ref{lem:ConvBinfty} are no longer valid.
Therefore, we need the following extension of a result in \cite[Lect.~21]{Levin1997_Lectures} which later leads to an interpolation formula for $\mathcal{B}^{\infty}_{T'/2}$.%

\begin{lemma}
\label{lem:LevinInterpol}
Let $S$ be a sine-type function of type $\widetilde{T}/2$ and let $\Lambda = \{\lambda_{n}\}_{n\in\ZN}$ be its zero set.
For any sequence $\bsc = \{c_{n}\}_{n\in\ZN} \in\ell^{\infty}$ there exists an entire function $g$ of exponential type $\widetilde{T}/2$ which solves the interpolation problem 
\begin{equation}
\label{equ:InterpolProbl}
	g(\lambda_{n}) = c_{n}\;,
	\quad n\in\ZN\;.
\end{equation}
Every entire function which fulfills condition~\eqref{equ:InterpolProbl} admits the representation
\begin{equation}
\label{equ:InterpolInfty}
	g(z)
	= {\sum_{n\in\ZN}}' c_{n}\, \frac{S(z)}{S'(\lambda_{n})} \left[\frac{1}{z-\lambda_{n}} + \frac{1}{\lambda_{n}}\right] + C_{s}\, S(z)\;,
\quad z\in\CN
\end{equation}
with an arbitrary constant $C_{s} \in \CN$ and where the sum converges uniformly and absolutely on compact subsets of $\CN$.
\end{lemma}
The prime at the summation sign in \eqref{equ:InterpolInfty} means that the second term in the braces is set to zero if $\lambda_{n} = 0$, and we will omit this prime subsequently.

\begin{remark}
For every $C_{s}\in\CN$, the function $g$ in \eqref{equ:InterpolInfty} has the property that
\begin{equation}
\label{equ:Remak_Conv}
	|g(\xi + \I\eta)|\, \E^{- \frac{\widetilde{T}}{2} |\eta|}
	= o(|\xi + \I\eta|)\;,
	\quad \text{as}\ |\xi + \I\eta|\to\infty\;.
\end{equation}
\end{remark}
\begin{remark}
Lemma~\ref{lem:LevinInterpol} states that for every $\bsc \in \ell^{\infty}$ the interpolation problem \eqref{equ:InterpolProbl} is uniquely solvable by an entire function $g$ of exponential type up to an additive sine-type term. 
But note that the function \eqref{equ:InterpolInfty} may not be bounded on $\RN$ for some sequences $\bsc \in \ell^{\infty}$, i.e. \eqref{equ:InterpolInfty} is not a $\mathcal{B}^{\infty}_{T'/2}$ function, in general.
\end{remark}

\begin{proof}
The proof is partly along the same lines as in \cite[Lect.~21, Theorem~1]{Levin1997_Lectures}.
First, we show that the sum in \eqref{equ:InterpolInfty} converges.
To this end, let $z \in \CN$ be arbitrary. Then there exists one $\lambda_{n_0} \in \Lambda$ such that
\begin{equation}
\label{equ:lambda_n0}
	\big| z - \lambda_{n_{0}}\big| \leq \big|z - \lambda_n \big|
	\quad\text{for all}\ n\in\ZN\;.
\end{equation}
Then we write the sum in \eqref{equ:InterpolInfty}  as
\begin{equation}
\label{equ:g}
	f(z) := S(z) {\sum_{n\neq n_{0}}} c_{n}\, \frac{1}{S'(\lambda_{n})} \frac{z}{(z-\lambda_{n}) \lambda_{n}}
	+ c_{n_0}\, \frac{S(z)}{S'(\lambda_{n_0})} \left[\frac{1}{z-\lambda_{n_0}} + \frac{1}{\lambda_{n_0}}\right]\;.
\end{equation}
Let $\varphi(z)$ be the sum on the left hand side. 
Since $S$ is a sine-type function, we have $\inf_{n\in\ZN} |S'(\lambda_{n})| =: C_{0} >0$ \cite{Young2001_Nonharmonic}, and Cauchy-Schwarz inequality gives 
\begin{align*}
	|\varphi(z)|	
	&\leq \frac{\|\bsc\|_{\infty}}{C_{0}} \sum_{\substack{n\neq n_0}} \frac{|z|}{|\lambda_{n}|\, |z - \lambda_{n}|}\\
	&\leq \frac{\|\bsc\|_{\infty}}{C_{0}}\, |z| \left( \sum_{n\neq n_0} \frac{1}{\left|\lambda_{n}\right|^2} \right)^{1/2} \left( \sum_{n\neq n_0} \frac{1}{|z - \lambda_{n}|^{2}} \right)^{1/2}.
\end{align*}
Since $\Lambda$ is the zero set of a function of exponential type, the series $\sum_{n\in\ZN} \left|\lambda_{n}\right|^{-2}$ converges to a finite constant $C_{1}$ \cite{Young2001_Nonharmonic}.
Moreover, because of \eqref{equ:lambda_n0}, we have for each $n \neq n_{0}$ that $|\lambda_{n} - \lambda_{n_0}| \leq |z - \lambda_{n}| + |z - \lambda_{n_0}| \leq 2\, |z - \lambda_n|$. Therewith one obtains
\begin{equation}
\label{equ:phiN_conv}
	|\varphi(z)|
	\leq \frac{2\, \|\bsc\|_{\infty}\, C_{1}}{C_{0}}\,  |z| \left( \sum_{n\neq n_0} \frac{1}{|\lambda_{n} - \lambda_{n_{0}}|^{2}} \right)^{1/2}\;.
\end{equation}
The sum on the right is convergent and can be upper bounded by a constant $C_{2}$ which is independent of $\lambda_{n_{0}}$, i.e. independent of $z$ (cf. Lemma~\ref{lem:SumZeros} in the appendix).

Consequently the sum on the right hand side of \eqref{equ:g} converges absolutely to an entire function $\varphi(z)$ uniformly on each compact subset of $\CN$.
Moreover, since $\lambda_{n_0}$ is a zero of $S$, the second term on the right hand side of \eqref{equ:g} defines an entire function whose modulus is upper bounded by $C_{3}\, |S(z)|$
for all $z$ with $|z - \lambda_{n_{0}}| > \delta$, and where $\delta > 0$ is an arbitrary constant and $C_{3} = C_{3}(\delta)$ a positive constant.
It then follows from \eqref{equ:phiN_conv} that $|f(z)| \leq C_{4} |S(z)|\, |z|$ for all $|z - \lambda_{n_{0}}| > \delta$, and
since $S$ is of exponential type $\widetilde{T}/2$, we see that $f$ satisfies \eqref{equ:Remak_Conv}.
Moreover, one can easily verify that $f$ solves $f(\lambda_{n}) = c_{n}$ for all $n\in\ZN$.

If $f$ is an entire function which satisfies \eqref{equ:InterpolProbl} and \eqref{equ:Remak_Conv} then also  $g(z) - f(z)$ satisfies \eqref{equ:InterpolProbl} and \eqref{equ:Remak_Conv}.
Since $S$ is of exponential type $\widetilde{T}/2$, it follows that the function $h(z) = [g(z) - f(z)]/S(z)$ satisfies $|h(z)| = o(|z|)$ as $|z| \to\infty$.
By Liouville's theorem $h(z)$ is a constant, which proves \eqref{equ:InterpolInfty}.
\end{proof}

The sequence $\bsc \in \ell^{\infty}$ in Lemma~\ref{lem:LevinInterpol} was arbitrary.
In our application however, it arises from sampling an entire function $v \in \mathcal{B}^{\infty}_{T'/2}$ so that the existence and boundedness of the solution of \eqref{equ:InterpolProbl} is naturally given.
The next lemma shows that once oversampling is applied, the additional term $C_{s}\, S(z)$ in the interpolation formula \eqref{equ:InterpolInfty} vanishes which allows a unique reconstruction of every $v \in \mathcal{B}^{\infty}_{T'/2}$ from its samples.

\begin{lemma}
\label{lem:InterpolSeriesInfty}
Let $S$ be a sine-type function of type $\widetilde{T}/2$ and let $\Lambda = \{\lambda_{n}\}_{n\in\ZN}$ be its zero set.
If $\widetilde{T}> T'$ then
\begin{equation}
\label{equ:Interpol2}
	v(z)
	= {\sum_{n\in\ZN}} v(\lambda_{n})\, \frac{S(z)}{S'(\lambda_{n})} \left[\frac{1}{z-\lambda_{n}} + \frac{1}{\lambda_{n}}\right]
	\qquad\text{for all}\ v \in \mathcal{B}^{\infty}_{T'/2}
\end{equation}
where the sum converges absolutely and uniformly on compact subsets of $\CN$.
Moreover, there exists a constant $C_{u}$ such that
\begin{equation}
\label{equ:globalApproxError}
	\lim_{N\to\infty} \max_{t\in\RN}
	\left| v(t) - \sum^{N}_{n=-N} v(\lambda_{n})\, \frac{S(t)}{S'(\lambda_{n})} \left[\frac{1}{t-\lambda_{n}} + \frac{1}{\lambda_{n}}\right]\right|
	\leq C_{u}\, \|v\|_{\mathcal{B}^{\infty}_{T'/2}}
\end{equation}
for all $v \in \mathcal{B}^{\infty}_{T'/2}$.
\end{lemma}

\begin{remark}
Relation \eqref{equ:globalApproxError} states that for every $v \in \mathcal{B}^{\infty}_{T'/2}$ the approximation error is uniformly bounded on the entire real axis $\RN$.
Then it follows from \eqref{equ:Interpol2} and \eqref{equ:globalApproxError} that the right hand side of \eqref{equ:Interpol2} converges to $v$ in the topology of $\mathcal{S}'$.
\end{remark}

\begin{remark}
The assumption that $\Lambda$ is the zero set of a sine-type function is only a sufficient condition which guarantees that $v \in \mathcal{B}^{\infty}_{T'/2}$ can be recovered from its samples at $\Lambda$.
Necessary and sufficient conditions on $\Lambda$ to be a sampling sequence for $\mathcal{B}^{\infty}_{T'/2}$ can be found in \cite{Ortega-Seip}. 
\end{remark}

\begin{proof}
Since $v \in \mathcal{B}^{\infty}_{T'/2}$ and $\Lambda$ is the zero set of a sine-type function,
the sequence $\{c_{n} = v(\lambda_{n})\}_{n\in\ZN}$ is in $\ell^{\infty}$.
Because $v \in \mathcal{B}^{\infty}_{T'/2} \subset \mathcal{B}^{\infty}_{\widetilde{T}/2}$, we know from Lemma~\ref{lem:LevinInterpol} that every entire function $g$ of exponential type $\widetilde{T}/2$ which satisfies $\widetilde{v}(\lambda_{n}) = c_{n}$ for all $n \in \ZN$ has the form $g(z) = f(z) + C_{s}\,S(z)$ where $f(z)$ stands for the sum on the right hand side of \eqref{equ:Interpol2}. Consequently, also $v$ has to have this form, i.e. $v(z) = f(z) + C_{s}\,S(z)$ and we have to prove that $C_{s} = 0$.
To this end, it is sufficient to show that $f$ is of exponential type $T'/2$. Then, also $v - f$ is of exponential type $T'/2$ such that in the equation
$v(z) - f(z) = C_{s}\, S(z)$
the modulus of the left hand side can be upper bounded by 
\begin{equation}
\label{HelpLemmaLevinReihe}
|v(z) - f(z)| \leq B_{1} \E^{\frac{T'}{2} |z|}
\end{equation}
for all sufficiently large $|z|$ and with a certain constant $B_{1}>0$.
Now assume $C_{s} > 0$. Because $S$ is of sine-type $\widetilde{T}/2$, we have $|v(z) - f(z)| = |C_{s} S(z)| \geq B_{2} \exp(\widetilde{T}/2\, |z|)$ for all sufficiently large $|z|$ with a constant $B_{2}>0$.
But since $\widetilde{T} > T'$ this yields a contradiction to the upper bound \eqref{HelpLemmaLevinReihe}.
It follows that $C_{s} = 0$, i.e. that \eqref{equ:Interpol2} holds.

It remains to show that $f$ is of exponential type $T'/2$. 
Without loss of generality, we assume $v(0) = 0$ (otherwise, one applies the following reasoning to $v(z) - v(0)$) and consider the function $w(z):=v(z)/z$. Since $v$ is of exponential type $T'/2$, for every $\epsilon > 0$ there exists a constant $C(\epsilon)$ such that
\begin{equation*}
	|w(z)| = \frac{1}{|z|}\, |v(z)|
	\leq \frac{1}{|z|}\, C(\epsilon)\, \E^{(\frac{T'}{2} + \epsilon) |z|}
	\leq C(\epsilon)\, \E^{(\frac{T'}{2} + \epsilon) |z|}\;
	\quad\text{for all}\ |z| \geq 1\;.
\end{equation*}
This shows that $w$ is of exponential type $T'/2$. Moreover the restriction of $w$ to $\RN$ is square integrable.
Therefore $w \in \mathcal{B}^{2}_{T'/2} \subset \mathcal{B}^{2}_{\widetilde{T}/2}$ and Lemma~\ref{lem:ConvBp} can be applied to $w$, which gives
\begin{equation*}
	w(z)
	= \frac{v(z)}{z}
	= \sum_{n\in\ZN} \frac{v(\lambda_{n})}{\lambda_{n}}\, \frac{S(z)}{S'(\lambda_{n})} \frac{1}{z-\lambda_{n}}\;.
\end{equation*}
By multiplying the whole equation by $z$, the right hand side becomes equal to $f(z)$ and shows that $f$ is of exponential type $T'/2$.

In order to prove \eqref{equ:globalApproxError}, we assume without loss of generality that $0 \notin \Lambda$ and we write $(\Op{T}_{N} v)(t)$ for the sum in \eqref{equ:globalApproxError},
and $(\Op{A}_{N} v)(t)$ for the finite series in Lemma~\ref{lem:ConvBp}, i.e.
\begin{equation*}
	(\Op{A}_{N} v)(t)
	= \sum^{N}_{n = -N} v(\lambda_{n})\, \frac{S(t)}{S'(\lambda_{n})} \frac{1}{t - \lambda_{n}}\;,
	\quad t\in\RN\;.
\end{equation*}
Both approximation series are obviously related by
\begin{equation*}
	(\Op{T}_{N} v)(t) = (\Op{A}_{N} v)(t) - \frac{(\Op{A}_{N} v)(0)}{S(0)}\, S(t)\;.
\end{equation*}
For  $v \in \mathcal{B}^{\infty}_{T'/2}$, the triangle inequality yields
\begin{equation*}
	\left| v(t) - (\Op{T}_{N} v)(t) \right|
	\leq \left| v(t) - (\Op{A}_{N}v)(t) \right| + \left| \frac{(\Op{A}_{N} v)(0)}{S(0)} \right|\, \left| S(t) \right|\;.
\end{equation*}
Now it is known \cite[Theorem~5]{MoenichBoche_SP2010} that the first term on the right hand side is uniformly bounded by $C_{1} \|v\|_{\infty}$ for all $t\in\RN$ and all $N\in\ZN$ with a certain constant $C_{1}$ independent of $t$ and $N$. The same result also implies that the second term is uniformly bounded by a constant of the form $C_{2}\, \|v\|_{\infty}$ which is why \eqref{equ:globalApproxError} holds and the proof is complete.
\end{proof}

Now we are ready to state a corollary of Theorem \ref{thm:MainThm}.
By adding an appropriate test signal $u$ prior to our sampling scheme (cf. Fig.~\ref{fig:MeasureSetup2}) we are able to ensure the ``non-zero requirement'' of Theorem \ref{thm:MainThm}, and therefore every signal in our signal space $\ourSpace^{p}$ can be reconstructed from magnitude measurements.

\begin{corollary}
\label{cor:Cor2}
For every $1 \leq p \leq \infty$ there exists a sine-type function $u$, and sets $\{ \alpha^{(m)}_{k} : k=1,\dots,K\;,\; m=1,\dots,K^{2}\}$ and $\{\lambda_{k}\}^{K}_{k=1}$ of modulation coefficients which satisfy condition A and B, respectively, such that every 
\begin{equation}
\label{ourLimitedSpace}
	x \in \{ x\in\ourSpace^{p}\ :\ \|x\| \leq 1 \}
\end{equation}
can be perfectly reconstructed (up to a constant phase) from the measurements
\begin{equation*}
	c^{(m)}_{n} = \left| \sum^{K}_{k=1} \alpha^{(m)}_{k}\, \widehat{(x + u)}(n\beta + \lambda_{k}) \right|^{2}\;,
	\quad n\in\ZN\ ,\ m=1,\dots,K^{2}\;.
\end{equation*}
\end{corollary}

\begin{remark}
The condition $\|x\| \leq 1$ in \eqref{ourLimitedSpace} only requires that an upper bound on the signal norm is known.
Practically, this is a very natural assumption and necessary to calibrate the measurement system by an appropriate amplitude of the additive test signal $u$.
\end{remark}

\begin{proof}
Let $x$ be as in \eqref{ourLimitedSpace}.
Because of \eqref{equ:boundedBp} there is an $M>0$ such that $|\Hx(\xi)| \leq M$ for all $\xi\in\RN$.
Since $\ourB^p \subset \ourB^{\infty}$ for all $p\in [1,\infty)$, it is sufficient to consider the case $p=\infty$.
Fix $T'>T$, then Lemma~\ref{lem:DufSchaef_2} shows that there exists a sine-type function $\Hu$ of type $T'/2$
and a constant $H$ such that the function
\begin{equation}
\label{equ:PlusU}
	\Hv(z) := \Hx(z) + \Hu(z)\;,
	\quad z\in\CN
\end{equation}
has no zeros for all $z = \xi + \I\eta\in\CN$ with $|\eta| > H$.
Choose $\Lambda^a = \{\lambda_{n} = \xi_{n} + \I\eta_{n}\}_{n\in \ZN}$ as the zero set of a sine-type function of type $\widetilde{T}/2 > T'/2$. 
By \cite{Levin_Perturbations} we can shift the imaginary parts of the points $\lambda_{n}$ such that $|\eta_n|>H$ for all $n$ while $\Lambda^a$ remains the zero set of a sine-type function.
Denote the corresponding sine-type function by $S$.

Now the signal \eqref{equ:PlusU} is modulated and sampled exactly as described in Sec.~\ref{sec:Measurement}.
Then the intensity measurements are given, similar as in \eqref{equ:IntensMeasure}, by $| \langle \Hbv_{n},\balpha_m \rangle|^{2}$.
Following the same steps as in the proof of Theorem~\ref{thm:MainThm}, one obtains the values of $\Hv$ at the sampling set $\Lambda^{a}$ up to a constant phase $\theta_{0}$.
Since by our construction overlapping interpolation points do not coincide with zeros of $\Hv$ the phase information can be propagated and we are able to recover $\Hv(\Lambda^a)\, \E^{\I \theta_0}$ from the intensity measurements for every signal $\Hv$ of the form \eqref{equ:PlusU}.

Since $\Hv \in \mathcal{B}^{\infty}_{T'/2}$ and $\sup_{n\in\ZN} |\eta_n|<H$, the sequence $\{ d_{n} := \Hv(\lambda_{n})\, \E^{\I\theta_0} \}_{n\in\ZN}$
is in $\ell^{\infty}$ such that Lemma~\ref{lem:InterpolSeriesInfty} can be applied to interpolate $\Hv(z)\, \E^{\I\theta_{0}}$ from these samples
\begin{equation*}
	\Hv(z)\,\E^{\I\theta_0} = \sum_{n\in\ZN}  d_{n}\, \frac{S(z)}{S'(\lambda_{n})} \left[\frac{1}{z-\lambda_{n}} + \frac{1}{\lambda_{n}}\right]\;,
	\quad z\in\CN\;.
\end{equation*}
However because $\theta_0$ is unknown, it is not possible to obtain $\Hx(z)$ directly from $\Hv(z)\,\E^{\I\theta_0}$ using \eqref{equ:PlusU}.
Instead, one can determine
\begin{align*}
  \widetilde{x}(z) &:= \Hv(z)\,\E^{\I\theta_0} - \Hu(z)
  = \Hx(z)\,\E^{\I\theta_0} - \Hu(z)\left( 1-\E^{\I\theta_0}\right)\;.
\end{align*}
If we choose $\Hu(z) = \cos(\frac{T'}{2}z)$ and apply the inverse Fourier-Laplace transform to $\widetilde{x}$, one obtains $x(t)\, \E^{\I\theta_0}$ for $t \in \TT$ since the inverse Fourier transform of the cosine function vanishes on $\TT$ (cf. Example~\ref{exa:Cosine}). 
\end{proof}

\section{Discussion and Outlook}
\label{sec:Summary}

To specify the sampling system in Fig.~\ref{fig:MeasureSetup}, one has to fix $K$, $M$, $a$ and $\beta$.
The number $K \geq 2$ can be chosen arbitrarily. Then $M = K^{2}$ is fixed, and $1\leq a\leq K-1$. 
The sampling period $\beta$ has to be chosen such that the sampling system satisfies Condition~B.
As discussed before, one possible choice may start with the zeros of the function $\sin(\frac{\widetilde{T}}{2} z)$ with $\widetilde{T}\geq T$.
Then $\delta := \lambda_{k} - \lambda_{k-1} = 2\pi/\widetilde{T}$ such that $\beta = (K-a)\,\delta$.
Therewith, the total sampling rate becomes
\begin{equation*}
R(a,K,\widetilde{T})
= \frac{M}{\beta} = \frac{K^2}{(K-a)\,\delta} = \frac{K^2}{K-a} \frac{\widetilde{T}}{2\pi} = \frac{K^2}{K-a}\frac{\widetilde{T}}{T} R_{\mathrm{Ny}}
\end{equation*} 
where $R_{\mathrm{Ny}} := T/(2\pi)$ is the Nyquist rate.
It is apparent that asymptotically, $R(a,K,\widetilde{T})$ grows proportionally to $K$, increases with the overlap $a$, and is bounded below by
\begin{equation*}
	\inf_{\substack{1\leq a<K, \\K\geq 1, \widetilde{T}>T}} R(a,K,\widetilde{T})
	= \inf_{\widetilde{T}>T}R(1,2,\widetilde{T})
	= 4 R_{\mathrm{Ny}}\,.
\end{equation*}
Since $\widetilde{T}/T$ can be made arbitrarily close to $1$ using Theorem~\ref{thm:MainThm} and Corollary~\ref{cor:Cor2}, we can sample at a rate which is almost as small as $4 R_{\mathrm{Ny}}$ while still ensuring perfect reconstruction. Asymptotically, this corresponds to the findings for finite dimensional spaces \cites{Balan1,Bodmann_StablePR2014,Fickus_VeryFewMeasurements}, where any $x \in \CN^{N}$ can be reconstructed from $M = O(4N)$  magnitude samples.

We note that the above framework can be applied exactly the same way for band-limited signals. To this end, one only has to exchange the time and frequency domain.
Then the modulators in Fig.~\ref{fig:MeasureSetup} have to be replaced by linear filters and the sampling of the magnitudes has to be done in the time domain.

Time-limited and band-limited signals are special cases of so-called modulation- and shift-invariant spaces, respectively.
Consequently, it is not surprising that, under some restriction on the generators of these space, it is fairly easy to extend the approach of the present paper to these signal spaces \cite{Pohl_ICASSP14}.

Moreover, using similar ideas as in this paper here, it is possible to construct deterministic measurement vectors for finite dimensional vector spaces $\CN^{N}$ which allow a very efficient and stable signal recovery from phaseless measurements \cite{PYB_STIP14}. More precisely, it is possible to construct a set $\{\psi_{n}\}^{4N-4}_{n=1}$ of measurement vectors in $\CN^{N}$ such that every $x = (x_{1},\dots,x_{N})^{\T} \in \CN^{N}$ with $x_{1} \neq 0$ can be reconstructed (up to a unitary factor) from the intensity measurements $\left|\left\langle x,\psi_{n}\right\rangle\right|^{2}$, $n=1,\dots,4N-4$.
Reconstruction is then based on a stable algebraic algorithm with a computational complexity which grows only linearly with the dimension~$N$.

\appendix
\section*{Appendix: An auxiliary lemma}
\renewcommand{\thesection}{A} 

\begin{lemma}
\label{lem:SumZeros}
Let $S$ be a sine-type function of type $\sigma$ and let $\Lambda = \{\lambda_n\}_{n\in\ZN}$ be its zero set.
If $\lambda_{n_{0}} \in \Lambda$ is an arbitrary zero of $S$ then
\begin{equation*}
	\sum_{n\neq n_{0}} \frac{1}{|\lambda_n - \lambda_{n_{0}}|^{2}} \leq C < \infty
\end{equation*}
with a constant $C$ which depends only on $S$ but not on $n_{0}$.
\end{lemma}

\begin{proof}
Since $S$ is a sine-type function there are constants $A,B,H$ such that inequalities \eqref{equ:SineTypeH} hold. Furthermore there exists a constant $\delta > 0$ such that $|\lambda_{m} - \lambda_{n}| \geq \delta$ for all $m\neq n$ and a constant $\alpha>0$ such that $|S'(\lambda_{n})| \geq \alpha$ for all $\lambda_{n} \in \Lambda$ (see, e.g., \cites{Levin1997_Lectures,Young2001_Nonharmonic}).
In particular, $S$ is an entire function of exponential type $\sigma$.
Then the Phragm\'en-Lindel{\"o}f principle and \eqref{equ:SineTypeH} imply that $S$ is bounded on every line parallel to $\RN$.
Therefore there exists a constant $M$ such that 
\begin{equation}
\label{equ:AppProof_1}
	|S(\xi + \I \eta)| \leq M\, \E^{\sigma |\eta|}\,
	\quad\text{for all}\ \xi, \eta\in \RN\;.
\end{equation}
Set $\widetilde{S}(z) := S(z + \lambda_{n_{0}})$ and write $\lambda_{n_{0}} = \xi_{0} + \I \eta_{0}$.
It follows from \eqref{equ:SineTypeH} that $|\eta_{0}| \leq H$ and \eqref{equ:AppProof_1} gives 
\begin{equation*}
	\big|\widetilde{S}(\xi + \I \eta)\big|
	\leq M\, \E^{\sigma |\eta + \eta_{0}|}
	\leq M\, \E^{\sigma H}\, \E^{\sigma |\eta|}
	\quad\text{for all}\ \xi, \eta\in \RN\;.
\end{equation*}
Consequently $|\widetilde{S}(z)| \leq \widetilde{M}\, \E^{\sigma |z|}$ for all $z\in\CN$ and with the constant $\widetilde{M} = M \E^{\sigma H}$ which depends only on $S$ but not on $n_{0}$.
The zeros of $\widetilde{S}$ are $\widetilde{\lambda}_{n} = \lambda_{n} - \lambda_{n_0}$, and
we assume that they are ordered increasingly by their absolute values,
i.e. such that $0 = |\widetilde{\lambda}_{0}| < \delta \leq |\widetilde{\lambda}_1| \leq |\widetilde{\lambda}_{2}| \leq \dots$.
Then we define $Q(z) := \widetilde{S}(z)/z$.
This is again an entire function of exponential type $\sigma$ which satisfies
\begin{equation}
\label{equ:App:PropQ}
\begin{split}
	&|Q(z)| \leq P\, \E^{\sigma |z|}\quad\text{for all}\ z\in\CN
	 \quad\text{and with}\ P = \widetilde{M}\, \E^{\sigma} = M\, \E^{\sigma(H+1)}\;,\\
	&\big| Q(0) \big| = \big| \widetilde{S}'(0) \big| = \big| \widetilde{S}'(\widetilde{\lambda}_{0}) \big| \geq \alpha\;,
\end{split}	
\end{equation}
and the zero set of $Q$ is obviously $\{\widetilde{\lambda}_n\}^{\infty}_{n=1}$.
If $n(r)$ denotes the number of zeros of $Q$ for which $|\widetilde{\lambda}_n| \leq r$, then Jensen's formula \cite{Levin1997_Lectures} and \eqref{equ:App:PropQ} imply
\begin{equation*}
	N(r)
	:= \int^{r}_{0} \frac{n(\tau)}{\tau}\, \d\tau
	= \frac{1}{2\pi} \int^{\pi}_{-\pi} \ln|Q(r\E^{\I\theta})|\, \d\theta - \ln|Q(0)|
	\leq \ln(P/\alpha) + \sigma\, r\;.
\end{equation*}
Since $n(r)$ is non-decreasing, we have $N(e r) \geq \int^{e r}_{r} n(\tau)\, \tau^{-1}\, \d\tau \geq n(r)$, where $e = \E^{1}$.
This yields the upper bound $n(r) \leq \ln(P/\alpha) + e\, \sigma\, r$.
If we take $r = |\widetilde{\lambda}_{n}|$ then $n(r) = n$ and one gets $n = n(r) \leq \ln(P/\alpha) + e\, \sigma\, |\widetilde{\lambda}_n|$.
Now we have
\begin{equation*}
	\sum_{n\neq n_{0}} \big|\lambda_n - \lambda_{n_{0}}\big|^{-2}
	= \sum^{\infty}_{n=1} \big|\widetilde{\lambda}_{n}\big|^{-2}
	= \sum^{N-1}_{n=1} \big|\widetilde{\lambda}_{n}\big|^{-2} + \sum^{\infty}_{n=N} \big|\widetilde{\lambda}_{n}\big|^{-2}
\end{equation*}
where $N \in \NN$ was chosen as the smallest integer such that $N \geq \ln(P/\alpha) + 1$.
Next we use for the first sum on the right hand side that $|\widetilde{\lambda}_{n}| \geq \delta$ for all $n \geq 1$. 
In the second sum we apply the bound $|\widetilde{\lambda}_{n}| \geq [ n-\ln(P/\alpha) ]/[e\, \sigma]$ from above. This gives
\begin{multline*}
	\sum_{n\neq n_{0}} \frac{1}{\big|\lambda_n - \lambda_{n_{0}}\big|^{2}}
	\leq \frac{N-1}{\delta^{2}} + \sum^{\infty}_{n=N} \frac{\E^{2} \sigma^{2}}{[n - \ln(P/\alpha)]^{2}}\\
	\leq \frac{\ln(P/\alpha) + 1}{\delta^{2}} + \E^{2} \sigma^{2} \sum^{\infty}_{m=1} \frac{1}{m^2}
	= \frac{\ln(e P / \alpha)}{\delta^{2}} + \frac{\E^{2} \pi^{2}}{6}\, \sigma^{2}
	=: C
\end{multline*}
where the constant $C<\infty$ is independent of $n_0$ and depends only on $S$.
\end{proof}

\paragraph{Acknowledgments}
This work was partly supported by the German Research Foundation (DFG) under Grant PO~1347/2-1 and BO~1734/22-1.



\begin{bibdiv}
\begin{biblist}

\bib{Balan_Painless}{article}{
      author={Balan, Radun},
      author={Bodmann, Bernhard~G.},
      author={Casazza, Peter~G.},
      author={Edidin, Dan},
       title={{Painless reconstruction from magnitudes of frame coefficients}},
        date={2009-08},
     journal={{J.~Fourier Anal.~Appl.}},
      volume={15},
      number={4},
       pages={488\ndash 501},
}

\bib{Balan1}{article}{
      author={Balan, Radun},
      author={Casazza, Peter~G.},
      author={Edidin, Dan},
       title={{On signal reconstruction without phase}},
        date={2006-05},
     journal={{Appl. Comput. Harmon. Anal.}},
      volume={20},
      number={3},
       pages={345\ndash 356},
}

\bib{Bauschke02}{article}{
      author={Bauschke, Heinz~H.},
      author={Combettes, Patrick~L.},
      author={Luke, D.~Russell},
       title={{Phase retrieval, error reduction algorithm, and Fienup variants:
  a view from convex optimization.}},
        date={2002-07},
     journal={{J.~Opt. Soc. Am.~A}},
      volume={19},
      number={7},
       pages={1334\ndash 1345.},
}

\bib{Boche_Pohl_IEEE_IT_InterpolatedData}{article}{
      author={Boche, Holger},
      author={Pohl, Volker},
       title={{On the calculation of the Hilbert transform from interpolated
  data}},
        date={2008-05},
     journal={{IEEE} Trans. Inf. Theory},
      volume={54},
      number={5},
       pages={2358\ndash 2366},
}

\bib{Bodmann_StablePR2014}{article}{
      author={Bodmann, Bernhard~G.},
      author={Hammen, Nathaniel},
       title={{Stable phase retrieval with low-redundancy frames}},
        date={2014-07},
     journal={{Adv. Compt. Math.}},
      volume={41},
      number={1},
        note={to appear},
}

\bib{Burge_PhaseProblem76}{article}{
      author={Burge, R.~E.},
      author={Fiddy, M.~A.},
      author={Greenaway, A.~H.},
      author={Ross, G.},
       title={{The phase problem}},
        date={1976-08},
     journal={Proc. R.~Soc. Lond.~A},
      volume={350},
      number={1661},
       pages={192\ndash 212},
}

\bib{CandesEldar_PhaseRetrieval}{article}{
      author={Cand{\`e}s, Emmanuel~J.},
      author={Eldar, Yonina~C.},
      author={Strohmer, Thomas},
      author={Voroninski, Vladislav},
       title={{Phase retrieval via matrix completion}},
        date={2013},
     journal={{SIAM J.~Imaging Sci.}},
      volume={6},
      number={1},
       pages={199\ndash 225},
}

\bib{DuffinSchaeffer_38}{article}{
      author={Duffin, Richard},
      author={Schaeffer, A.~C.},
       title={{Some properties of functions of exponential type}},
        date={1938-04},
     journal={{Bull. Amer. Math. Soc.}},
      volume={44},
       pages={236\ndash 240},
}

\bib{Falldorf_SLM10}{article}{
      author={Falldorf, Claas},
      author={Agour, Mostafa},
      author={v.~Kopylow, Christoph},
      author={Bergmann, Ralf~B.},
       title={{Phase retrieval by means of spatial light modulator in the
  Fourier domain of an imaging system}},
        date={2010-04},
     journal={Applied Optics},
      volume={49},
      number={10},
       pages={1826\ndash 1830},
}

\bib{Fickus_VeryFewMeasurements}{article}{
      author={Fickus, Matthew},
      author={Mixon, Dustin~G.},
      author={Nelson, Aaron~A.},
      author={Wang, Yang},
       title={{Phase retrieval from very few measurements}},
        date={2014-05},
     journal={Linear Algebra Appl.},
      volume={449},
       pages={475\ndash 499},
}

\bib{Fienup1982_PhaseRetrieval}{article}{
      author={Fienup, J.~R.},
       title={{Phase retrieval algorithms: a comparison}},
        date={1982-08},
     journal={{Applied Optics}},
      volume={21},
      number={15},
       pages={2758\ndash 2769},
}

\bib{Fienup_93}{article}{
      author={Fienup, J.~R.},
      author={Marron, J.~C.},
      author={Schulz, T.~J.},
      author={Seldin, J.~H.},
       title={Hubble space telescope characterized by using phase-retrieval
  algorithms},
        date={1993-04},
     journal={Appl. Opt.},
      volume={32},
      number={10},
       pages={1747\ndash 1767},
}

\bib{Finkelstein_QuantumCom04}{article}{
      author={Finkelstein, J.},
       title={{Pure-state informationally complete and "really" complete
  measurements}},
        date={2004},
     journal={Phys. Rev.~A},
      volume={70},
       pages={052107},
}

\bib{Goodman_FourierOptik}{book}{
      author={Goodman, Joseph~W.},
       title={{Introduction to Fourier optics}},
   publisher={McGraw-Hill Comp.},
     address={New York},
        date={1996},
}

\bib{Oppenheim_Phase_80}{article}{
      author={Hayes, Monson~H.},
      author={Lim, Jae~S.},
      author={Oppenheim, Alan~V.},
       title={{Signal reconstruction from phase or magnitude}},
        date={1980-12},
     journal={{IEEE} Trans. Acoust., Speech, Signal Process.},
      volume={ASSP-28},
      number={6},
       pages={672\ndash 680},
}

\bib{Hoermander_LinDifOp}{book}{
      author={H{\"o}rmander, Lars},
       title={{Linear partial differential operators}},
   publisher={Springer-Verlag},
     address={Berlin},
        date={1976},
}

\bib{Jaming_Radar10}{inproceedings}{
      author={Jaming, Philippe},
       title={{The phase retrieval problem for the radar ambiguity function and
  vice versa}},
        date={2010-05},
   booktitle={{IEEE Intern. Radar Conf.}},
     address={Washington, DC, USA},
}

\bib{KatkovnikAstola_12}{article}{
      author={Katkovnik, Vladimir},
      author={Astola, Jaakkoo},
       title={{Phase retrieval via spatial light modulators phase modulation in
  4f optical setup: numerical inverse imaging with sparse regularization for
  phase and amplitude}},
        date={2012-01},
     journal={J.~Opt.~Soc.~Amer. A},
      volume={29},
      number={1},
       pages={105\ndash 116},
}

\bib{Levenshtein_98}{article}{
      author={Levenshtein, Vladimir},
       title={On designs in compact metric spaces and a universal bound on
  their size},
        date={1998},
     journal={Discrete Math.},
      volume={192},
       pages={251\ndash 271},
}

\bib{Levin1997_Lectures}{book}{
      author={Levin, B.~Y.},
       title={Lectures on entire functions},
   publisher={American Mathematical Society},
     address={Providence, RI},
        date={1997},
}

\bib{Levin_Perturbations}{article}{
      author={Levin, B.~Y.},
      author={Ostrovskii, I.~V.},
       title={Small perturbations of the set of roots of sine-type functions},
        date={1979},
     journal={Izv. Akad. Nauk SSSR Ser. Mat},
      volume={43},
      number={1},
       pages={87\ndash 110},
}

\bib{Vetterli2011_PRbeyondFienup}{inproceedings}{
      author={Lu, Yue~M.},
      author={Vetterli, Martin},
       title={{Sparse spectral factorization: unicity and reconstruction
  algorithms}},
        date={2011-05},
   booktitle={{Proc. 36th Intern. Conf. on Acoustics, Speech, and Signal
  Processing (ICASSP)}},
     address={Prague, Czech Republic},
       pages={5976\ndash 5979},
}

\bib{Marchesini_07}{article}{
      author={Marchesini, Stefano},
       title={{Phase retrieval and saddle-point optimization}},
        date={2007-10},
     journal={J.~Opt.~Soc.~Amer. A},
      volume={24},
      number={10},
       pages={3289\ndash 3296},
}

\bib{Miao_XRay_08}{article}{
      author={Miao, Jianwei},
      author={Ishikawa, Tetsuya},
      author={Shen, Qun},
      author={Earnest, Thomas},
       title={{Extending X-ray crystallography to allow the imaging of
  noncrystalline materials, cells, and single protein complexes}},
        date={2008-05},
     journal={Annu. Rev. Phys. Chem.},
      volume={59},
       pages={387\ndash 410},
}

\bib{Millane_90}{article}{
      author={Millane, R.~P.},
       title={{Phase retrieval in crystallography and optics}},
        date={1990-03},
     journal={J.~Opt.~Soc.~Amer. A},
      volume={7},
      number={3},
       pages={394\ndash 411},
}

\bib{MoenichBoche_SP2010}{article}{
      author={M{\"o}nich, Ullrich~J.},
      author={Boche, Holger},
       title={Non-equidistant sampling for bounded bandlimited signals},
        date={2010-07},
     journal={Signal Processing},
      volume={90},
      number={7},
       pages={2212\ndash 2218},
}

\bib{Ortega-Seip}{article}{
      author={Ortega-Cerd{\`a}, Joaquim},
      author={Seip, Kristian},
       title={{Multipliers for Entire Functions and an Interpolation Problem of
  Beurling}},
        date={1999},
     journal={J. Funct. Anal.},
      volume={162},
      number={2},
       pages={400\ndash 415},
}

\bib{PYB_STIP14}{article}{
      author={Pohl, Volker},
      author={Yang, Fanny},
      author={Boche, Holger},
       title={{Phase retrieval from low-rate samples}},
        date={2014-05},
     journal={{Sampl. Theory Signal Image Process.}},
      volume={13},
      eprint={arXiv:1311.7045},
        note={to appear},
}

\bib{Pohl_ICASSP14}{inproceedings}{
      author={Pohl, Volker},
      author={Yapar, C.},
      author={Boche, Holger},
      author={Yang, Fanny},
       title={{A phase retrieval method for signals in modulation-invariant
  spaces}},
        date={2014-05},
   booktitle={{Proc. 39th Intern. Conf. on Acoustics, Speech, and Signal
  Processing (ICASSP)}},
     address={Florence, Italy},
}

\bib{Rabiner_SpeechRec}{book}{
      author={Rabiner, Lawrence},
      author={Juang, Biing-Hwang},
       title={{Fundamentals of speech recognition}},
   publisher={Prentice Hall, Inc.},
     address={Englewood Cliffs},
        date={1993},
}

\bib{Ross_PhaseProblem78}{article}{
      author={Ross, G.},
      author={Fiddy, M.~A.},
      author={Nieto-Vesperinas, M.},
      author={Wheeler, M.~W.~L.},
       title={{The phase problem in scattering phenomena: The zeros of entire
  functions and their significance}},
        date={1978-03},
     journal={Proc. R.~Soc. Lond.~A},
      volume={360},
      number={1700},
       pages={25\ndash 45},
}

\bib{Rudin_FktAnalysis}{book}{
      author={Rudin, Walter},
       title={{Functional analysis}},
     edition={2},
   publisher={McGraw-Hill},
     address={Boston},
        date={1991},
}

\bib{Thakur2011}{article}{
      author={Thakur, Gaurav},
       title={Reconstruction of bandlimited functions from unsigned samples},
        date={2011-08},
     journal={J.~Fourier Anal.~Appl.},
      volume={17},
      number={4},
       pages={720\ndash 732},
}

\bib{Vladimirov}{book}{
      author={Vladimirov, V.~S.},
       title={{Equations of mathematical physics}},
   publisher={Marcel Dekker, Inc.},
     address={New York},
        date={1971},
}

\bib{Xiao_DistortedObject05}{article}{
      author={Xiao, Xianghui},
      author={Shen, Qun},
       title={{Wave propagation and phase retrieval in Fresnel diffraction by a
  distorted-object approach}},
        date={2005},
     journal={Phys. Rev.~B},
      volume={72},
       pages={033103},
}

\bib{Yang_SampTA13}{inproceedings}{
      author={Yang, Fanny},
      author={Pohl, Volker},
      author={Boche, Holger},
       title={{Phase retrieval via structured modulations in Paley-Wiener
  spaces}},
        date={2013-07},
   booktitle={{Proc. 10th Intern. Conf. on Sampling Theory and Applications
  (SampTA)}},
     address={Bremen, Germany},
}

\bib{Young2001_Nonharmonic}{book}{
      author={Young, Robert~M.},
       title={An introduction to nonharmonic fourier series},
   publisher={Academic Press},
     address={Cambridge},
        date={2001},
}

\bib{Zauner_Quantendesigns}{article}{
      author={Zauner, Gerhard},
       title={{Quantum designs: Foundations of a noncommutative design
  theory}},
        date={2011-02},
     journal={{Int. J. Quantum Inform.}},
      volume={9},
      number={1},
       pages={445\ndash 507},
}

\bib{Zhang_ApatureMod07}{article}{
      author={Zhang, Fucai},
      author={Pedrini, Giancarlo},
      author={Osten, Wolfgang},
       title={{Phase retrieval of arbitrary complex-valued fields through
  aperture-plane modulation}},
        date={2007},
     journal={Phys. Rev.~A},
      volume={75},
       pages={043805},
}

\end{biblist}
\end{bibdiv}
\end{document}